\documentclass[journal,12pt,onecolumn,draftclsnofoot]{IEEEtran}

 \usepackage[usenames,dvipsnames]{pstricks}
 \usepackage{epsfig}
 \usepackage{pst-grad} 
 \usepackage{pst-plot} 

\usepackage{graphicx}
  \usepackage{amsmath,amsthm}
  \usepackage{stfloats}
  \usepackage{mathdots}
  \usepackage{amssymb,bm}
  \usepackage{amsfonts}
  \usepackage{amscd}
  \usepackage{cite}
  \usepackage{color}
  \usepackage{dsfont}
  \usepackage{textcomp}
  \usepackage{etoolbox}
\usepackage{verbatim} 
\usepackage{lettrine}
   \usepackage{algorithm}
   \usepackage{algorithmicx}
   \usepackage{algpseudocode}
  \usepackage{enumerate} 
  \usepackage[font={it,rm}]{caption} 

  \usepackage{subcaption}
\usepackage{url}
\usepackage{tikz}
\usetikzlibrary{decorations.pathreplacing,angles,quotes,shapes}

\DeclareMathOperator{\argmin}{arg\,min}

\newtheorem{lemma}{Lemma}
  \newtheorem{theorem}{Theorem}
  \newtheorem{prop}{Proposition}
  \newtheorem{cor}{Corollary}[theorem]

\theoremstyle{remark}
\newtheorem{remark}{Remark}

\newtoggle{SINGLE_COL}
\toggletrue{SINGLE_COL}
\togglefalse{SINGLE_COL}

\usepackage{setspace}
\AtBeginEnvironment{algorithm}{\singlespacing}

\IEEEoverridecommandlockouts 

\title{Caching under Content Freshness Constraints}

\author{Pawan Poojary, Sharayu Moharir and Krishna Jagannathan
\thanks{P.~Poojary and K.~Jagannathan are with the Department of Electrical Engineering, IIT Madras, Chennai 600036, India. Email: \{ee15s025,krishnaj\}@ee.iitm.ac.in}%
\thanks{S.~Moharir is with the Department of Electrical Engineering, IIT Bombay, Mumbai 400076, India. Email: sharayum@ee.iitb.ac.in}%
\thanks{Some preliminary results contained in this work will be presented as a poster at COMSNETS 2018, Bengaluru, India.}
}

\begin{document}

\maketitle
\vspace{-11mm}
\begin{abstract}
Several real-time delay-sensitive applications pose varying degrees of freshness demands on the requested content. The performance of cache replacement policies that are agnostic to these demands is likely to be sub-optimal. Motivated by this concern, in this paper, we study caching policies under a request arrival process which incorporates user freshness demands. We consider the performance metric to be the steady-state cache hit probability. We first provide a universal upper bound on the performance of any caching policy. We then analytically obtain the content-wise hit-rates for the Least Popular (LP) policy and provide sufficient conditions for the asymptotic optimality of cache performance under this policy. Next, we obtain an accurate approximation for the LRU hit-rates in the regime of large content population. To this end, we map the characteristic time of a content in the LRU policy to the classical Coupon Collector's Problem and show that it sharply concentrates around its mean. Further, we develop modified versions of these policies which eject cache redundancies present in the form of stale contents. Finally, we propose a new policy which outperforms the above policies by explicitly using freshness specifications of user requests to prioritize among the cached contents. We corroborate our analytical insights with extensive simulations.
\end{abstract}


\section{Introduction} 
\label{sec1}

\IEEEPARstart{T}{he} past few years have seen a tremendous increase in real-time delay-sensitive applications that are tasked with monitoring processes, optimizing process parameters to achieve specific objectives and data logging for learning and prediction \cite{yick}. In this scenario, the network constitutes a sensor grid where several smart devices connect wirelessly to a single Access Point (AP) which serves these applications with requested contents.


A key characteristic of real-time applications is the varying degree of constraints that they impose on the freshness of the contents delivered by the network \cite{shenker}. The freshness of a content refers to the amount of time elapsed since it was last fetched from the back-end server. These freshness constraints are determined by a number of factors, including the nature of the data that a content represents, and the application for which it is used. For example, a data-logging application requesting data from a sensor grid that measures the pollution levels in a city can tolerate delays of the order of several minutes. However, an emergency alarm system requesting data from sensors that measure the concentration of noxious gases in a chemical factory cannot tolerate delays beyond a few seconds. Therefore, when dealing with delay-sensitve applications, the network should be designed to take into account such varied freshness demands.

Distributed content caching has been employed in networks on a large scale owing to its well-known advantages \cite{melamed}. Caching contents close to the end-users in a network serves to minimize the load on the network back-end by transferring it towards the network end-nodes. In our scenario, this translates to the contents being cached in the AP. The reduction in the back-end traffic saves precious network bandwidth and decreases content-delivery latency. However, owing to physical resource constraints, caches are capable of storing only a small fraction of the entire content catalogue. This is further exacerbated by the ever-increasing content population. Hence, optimal cache replacement strategies aimed at minimizing the number of deferred requests must be devised. Traditionally, caching strategies have mainly focussed on caching relatively popular contents.


In light of the above needs, it appears reasonable that existing caching policies should be designed by taking both content popularities and user freshness demands into consideration while taking the content storage and replacement decisions. Motivated by these considerations, we first analyze traditional cache replacement policies under user freshness demands. Secondly, we devise new policies that account for the relative content popularities as well as the user freshness constraints.

\vspace{-3mm}
\subsection{Related Work} \label{Related_Work}
Several analytical models have been developed to estimate hit-rates of popular caching policies \cite{Breslau}, \cite{King}, \cite{Serpanos2, Serpanos3, Flajolet, Dan:1990:AAL:98460.98525, Che2002, Fricker, Wolman, Gomaa, Mookerjee:2002:ALR:767824.769479}. Among these, there exists a class of policies that prioritize contents by estimating their popularity and storing the most popular contents in the cache. Several methods have been proposed in the literature to estimate the content popularity \cite{Good_Turing}, \cite{McAllester}, \cite{valiant}. However, all of them obtain the exact popularity distribution asymptotically and therefore behave identically in the steady state. We refer to policies in this class as the Least Popular (LP) policy. Perfect Least Frequently Used (Perfect-LFU) is the most basic among these policies \cite{Serpanos2}, \cite{Serpanos3}. LP policy provides the best performance in steady-state but at the cost of a large memory (equal to number of content types). On the other hand, the Least Recently Used (LRU) policy which although sub-optimal requires much smaller memory (equal to cache-size) and hence is widely-used.

Exact Markov chain analysis for LRU and First-In-First-Out (FIFO) policies was provided in \cite{King}, resulting in hit-rate expressions with exponential complexity whereas \cite{Flajolet} provided simpler expressions for the same. Dan \emph{et al.} \cite{Dan:1990:AAL:98460.98525} obtained  approximate hit-rates with much lower computational complexity, and showed that LRU outperforms FIFO. Che \emph{et al.} \cite{Che2002} proposed a simple approach to accurately estimate hit-rates for LRU policy under certain approximations and came up with the notion of a ``characteristic time" of the cache. The accuracy of the hit-rates and scope of the approximation were further investigated in \cite{Fricker}. A study by \cite{KamISIT} proposed a dynamic model for content requests that depends on both the freshness and popularity of contents. However, none of the above studies consider content requests having freshness demands which have to be met by the content delivery system.

Studies in \cite{Wolman, Gomaa, Mookerjee:2002:ALR:767824.769479} address the notion of freshness by considering that an object entering the cache has a limited lifetime beyond which it is expires. The object lifetimes are considered to be exponentially distributed in \cite{Wolman} and \cite{Gomaa}. In contrast to \cite{Wolman} which considers infinite cache, we obtain hit-rates under finite caching resources. The study by Mookerjee \emph{et al.} \cite{Mookerjee:2002:ALR:767824.769479} on LRU policy considers requests to a fixed number of documents that get updated periodically with a certain frequency. Whereas in our model, a content fetch from the source at any instant gives the latest version of the requested object \emph{i.e.}, the source produces real-time content. The main difference between our analytical model and the models in \cite{Wolman}, \cite{Gomaa}, \cite{Mookerjee:2002:ALR:767824.769479} is that in our case, it is the content request stream which demands that a content not older than a particular age be served. Also, it is upto the caching policy to decide how long a content resides in the cache.

\vspace{-3mm}
\subsection{Our Contributions}  
We modify the well-known Independent Reference Model (IRM) of request arrivals to incorporate freshness demands of users \cite{Breslau}, \cite{King}, \cite{Podlipnig:2003:SWC:954339.954341}. Under this traffic model, we evaluate cache performance in terms of its steady-state hit probability which is defined as the probability that a request will be served by the cache. We highlight the contributions of our work as follows:
\begin{itemize}
\item[$(i)$] We derive a universal upper bound on the performance of all caching policies subject to user freshness constraints.


\item[$(ii)$] We derive content-wise hit-rates for the LP policy with freshness constraints. For Zipf distributed content requests, we prove that as the number of contents ($n$) increases, the LP policy is asymptotically optimal, \emph{i.e.}, it attains upper bound performance as long as the cache-size increases with $n$.

\item[$(iii)$] We conduct an asymptotic analysis of the LRU policy, and obtain an approximation for its content-wise hit-rates. Our analysis is based on relating the \emph{characteristic time} \cite{Che2002} of a content in the LRU policy to the classical Coupon Collector's Problem \cite{ferrante}. Specifically, we show that this characteristic time enjoys tight concentration about its mean for large $n$, which then allows us to obtain an accurate approximation to the LRU hit-rates. We believe that our asymptotic analysis of the LRU policy is of independent interest, regardless of any freshness considerations.


\item[$(iv)$] We consider a class of policies that explicitly use the freshness specifications to eject \emph{stale} contents from the cache. In particular, we develop improved versions of the LP and LRU policies. Further, we propose a new policy which accounts for freshness specifications as well as content popularities in its cache replacement strategy. We conduct extensive simulations to corroborate our claim that this policy outperforms the other policies considered.
\end{itemize}

\vspace{-5mm}
\subsection{Organisation}
The remainder of the paper is organized as follows. Firstly, we describe our system model in Section \ref{system_model}. Metrics for evaluation of cache performance are defined in Section \ref{metric}. In Section \ref{agnostic}, we characterize the class of policies that are unaware of user freshness demands and in particular analyze the performance of LP and LRU policies. Analytical arguments to justify the approximations made while evaluating the LRU policy performance have been provided in Section \ref{no_name}. The class of policies that evict stale contents from the cache are studied in Section \ref{aware}. We provide proofs for the main results obtained in this paper in Section \ref{main_results}. In Section \ref{simulation}, we present the simulation results. Finally, the conclusions obtained from our analysis are presented in Section \ref{conclusion}.

\section{System Model}  \label{system_model}
In our model, we consider a front-end server with a finite cache. It serves incoming client requests with contents fetched from a back-end.

\begin{figure}
\centering
\scalebox{0.85} 
{
\vspace{2mm}
{
\begin{tikzpicture}[scale=0.35]
\draw (0,0) -- (-10.25,0);
\draw (0,0) -- (0,5);
\draw (0,5) -- (6.5,5);
\draw (6.5,5) -- (6.5,-5);
\draw (6.5,-5) -- (0,-5);
\draw (0,-5) -- (0,0);
\draw (1.5,-3.25) -- (1.5,3.25);
\draw (1.5,3.25) -- (5,3.25);
\draw (5,3.25) -- (5,-3.25);
\draw (5,-3.25) -- (1.5,-3.25);
\draw (-3.25,0.25) -- (-2.75,-0.25);
\draw (-3.25,-0.25) -- (-2.75,0.25);
\draw (-5.25,0.25) -- (-4.75,-0.25);
\draw (-5.25,-0.25) -- (-4.75,0.25);
\draw (-7.25,0.25) -- (-6.75,-0.25);
\draw (-7.25,-0.25) -- (-6.75,0.25);
\draw (-9.25,0.25) -- (-8.75,-0.25);
\draw (-9.25,-0.25) -- (-8.75,0.25);
\draw (1.5,2.25) -- (5,2.25);
\draw (1.5,1.25) -- (5,1.25);
\draw [loosely dotted,very thick] (3.25,1) -- (3.25,-0.5);
\draw [dashed] (6.5,0) -- (13,3.6);
\draw [dashed] (6.5,0) -- (13,-3.6);
  \node at (15,1.5) [rectangle,draw] (v100){};
  \node at (20,1.5) [rectangle,draw] (v100){};
  \node at (18,0) [rectangle,draw] (v100){};
  \node at (14,-1.5) [rectangle,draw] (v100){};
  \node at (21,-1.5) [rectangle,draw] (v100){};
\node at (18,0) [cloud, draw,cloud puffs=10,cloud puff arc=120, cloud ignores aspect, minimum width=5cm, minimum height=3.5cm] {};
%
\node at (3.25,-1.75) {$m$};
\node at (3.25,-4.25) {Cache};
\node at (3.25,-6.25) {Access Point};
\node at (3.25,-8) {(Front-end Server)};
\node at (18,-6.25) {Sensor Network};
\node at (18,-8) {(Back-end)};
\node at (-7,2) {Request arrival process};
\node at (-6.75,-1.25) {\begin{small}$X_{k}$\end{small}};
\node at (-4.40,-1.25) {\begin{small}$X_{k+1}$\end{small}};
\node at (16,2) {\begin{small}$1$\end{small}};
\node at (21,2) {\begin{small}$2$\end{small}};
\node at (22,-1) {\begin{small}$n$\end{small}};

\end{tikzpicture}
} }
\caption{Content delivery to incoming requests through a front-end server equipped with finite caching resources.}  \label{system_model}
\end{figure}
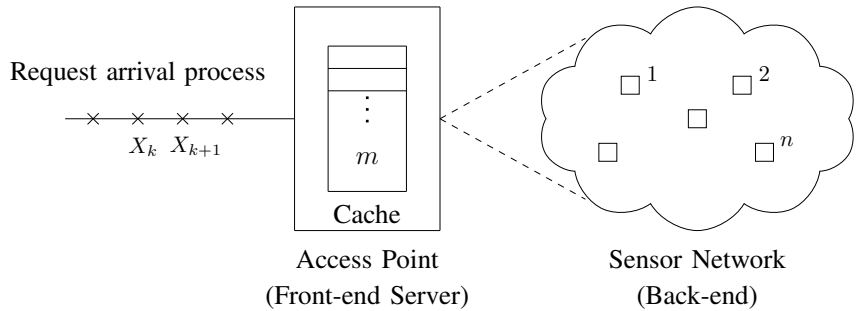


\subsection{Server and Storage Model}
The system consists of an Access Point (AP) equipped with a cache of length \textit{m}. The AP fetches contents from a population of \textit{n} content-generating objects indexed by $\lbrace 1,2,\ldots,n \rbrace$ and stores them in the cache in order to serve future requests. The object population could constitute a sensor grid where several smart-devices connect wirelessly to the AP. Each content fetched from the objects is of unit size and occupies unit space in the cache. Owing to practical resource constraints, typically $ m \ll n. $ The AP acts as a front-end server and is assigned the task of serving incoming content requests either from its cache or by fetching contents from the back-end.  

\vspace{-5mm}
\subsection{Request Arrival Model}
We adopt the Independent Reference Model (IRM) which is known to be a well suited abstraction for independent requests generated from a large population of users \cite{Breslau}, \cite{King}, \cite{Podlipnig:2003:SWC:954339.954341}. The request arrival process is modeled as an infinite sequence of independent and identically distributed random variables $\lbrace X_1,X_2,\ldots \rbrace$, where $X_k=i$ denotes that the $k^{th}$ request is for content type $i \in \lbrace1,2,\ldots,n \rbrace$. We denote popularity of content \textit{i} by $ p_i \triangleq \mathbb{P}(X_k=i) $. Typically, in a time slot $k$, let $X_k=i$. This request is accompanied by a freshness specification $F(i)$, which denotes the maximum acceptable \emph{age} of content $i$ that can still serve the request. When the AP requests a content $i$ from the back-end, object $i$ produces a new content (\emph{i.e.}, latest data sensed) and transmits it to the AP. Hence, a content currently fetched by the AP has zero age.

\begin{remark}
The \emph{age} of a cached content refers to the number of time slots elapsed since it was fetched from the back-end. \label{rem_age}
\end{remark}

\begin{remark}
Several studies, using empirical data from traces of web proxy caches, have shown that content request distributions follow the Zipf's law with varying exponents \cite{Breslau}, \cite{cunha}. Hence, for the asymptotic analysis of caching policies and numerical simulations, we assume Zipf distributed requests. Under the Zipf's law, $ p_i \propto i^{-\beta}$, where $\beta > 0$ is known as the Zipf parameter.
\label{rem_zipf}
\end{remark}

\vspace{-5mm}
\subsection{Service Model}
The server serves the incoming requests in the following manner:
\begin{itemize}
\item If the content corresponding to a request is not present in the cache, a cache \emph{miss} occurs.
\item If the requested content is present in the cache, its age is compared with the freshness specification of the request. If this age exceeds the freshness specification, the cached content is considered \textit{stale}; a cache miss occurs. On the other hand, if this age is within the freshness limits, the content is considered \textit{fresh}; a cache \emph{hit} occurs and the request is served.
\item In the event of a cache miss, the server has to fetch the content from the back-end and serve the request. At this point, the server has to decide whether or not to retain the fetched content in the cache.
\end{itemize}

\begin{remark}
\normalfont A content fetch is initiated if and only if there is a cache miss.
\label{rem1}
\end{remark}

\begin{remark}
\normalfont A request is always served with a \textit{fresh} content.
\label{rem2}
\end{remark}

\vspace{-5mm}
\section{System Performance Metrics}   \label{metric}
The hit-rate or hit-ratio of object $i$ at time $t$, $H_\mathcal{A}(i,t)$ is the ratio of the number of cache hits to the total number of requests for object $i$ received till time $t$ under a policy $\mathcal{A} \in \mathbb{A}$. Here, $\mathbb{A}$ is the set of all replacement policies. The steady-state hit-rate $h_\mathcal{A}(i)= \underset{t \rightarrow \infty}{\lim}H_\mathcal{A}(i,t)$\footnote{The limit exists since the caching process constitutes an ergodic Markov chain, that is, it is irreducible and aperiodic.} is the steady-state probability that a \textit{fresh} content $i$ is present in the cache, simply referred to as the \textit{hit rate} of object $i$ under policy $\mathcal{A}$. This follows from the fact that the caching process is ergodic\footnote{A process whose time average converges to its ensemble average.}.
\par However, the overall cache performance is quantified by the total probability of a cache hit in the steady state, simply referred to as the \emph{hit probability}. Let $ Y_k^i$ denote the event that a fresh content $i$ is present in the cache in the $k^{\text{th}}$ time-slot. Then, the hit probability in a discrete time slot for a policy $\mathcal{A} \in \mathbb{A}$ is given by
\begin{align}
\mathbb{P}_\mathcal{A}(\text{hit}) & \triangleq \sum_{i=1}^{n} \mathbb{P} \big( \textit{cache hit} \: \vert \: X_k=i \big) \: \mathbb{P} (  X_k=i ) \nonumber \\
& \overset{(a)}{=} \sum_{i=1}^{n} \mathbb{P} \big(  Y_k^i \: \vert \: X_k=i \big) \: \mathbb{P}( X_k=i ) \overset{(b)}{=} \sum_{i=1}^{n} \mathbb{P} ( Y_k^i ) \: \mathbb{P}( X_k=i ) = \sum_{i=1}^{n} h_\mathcal{A}(i) \, p_{i}. 
\end{align}
Step ($a$) holds due to the fact that, under the condition that content $i$ is requested, a cache hit is analogous to a fresh content $i$ being present in the cache. Step ($b$) follows from the fact that, in the $k^{\text{th}}$ time-slot, the presence of a fresh content in the cache is independent of any ensuing request for it.

\section{Performance of traditional freshness-agnostic caching policies} \label{agnostic}  
A replacement policy decides the strategy used to replace an existing content in the cache with the fetched content. We revisit the class of traditional caching policies $ \mathbb{T} $ studied in the literature under the constraints imposed by the freshness demands. These policies are inherently indifferent to the presence of \textit{stale} data in the cache. 

Let $\lbrace X_1,X_2,\ldots \rbrace$ be the request stream, $ C_i $ refer to content $i$, \emph{Age}$(C_i)$ refer to the number of time-slots for which content $i$ has resided in the cache and $ F(i)$ be its freshness specification, $ i\in \lbrace1,2,\ldots,n\rbrace $. Then, a policy $\mathcal{T} \in \mathbb{T}$ is implemented as per Algorithm \ref{firstalgo}.
 
\begin{algorithm} 
\caption{Freshness-agnostic caching.} \label{firstalgo}
\begin{algorithmic}[1]
\State $k \gets 1$.
\Loop
\If { $C_{X_k} \in \text{cache}$ }
     \If {$Age(C_{X_k}) \leq F(X_k)$}
         \State serve the request \Comment{cache hit}
     \Else    
         \State Fetch content type $C_{X_k}$ from back-end 
         \State Replace old copy of $C_{X_k}$ \Comment{cache miss}
     \EndIf    
\Else
    \State Fetch content type $C_{X_k}$ from back-end
    \State serve the request  \Comment{cache miss}
    \If {cache is full}
        \State Implement cache replacement policy $ \mathcal{T} $ \label{Algo1step14}
         \State {\textbf{if} $C_{X_k} \in \text{cache}$ \textbf{then} $Age(C_{X_k}) \gets 1$ }
     \Else   
         \State Place $C_{X_k}$ in cache. $Age(C_{X_k}) \gets 1$
     \EndIf     
\EndIf     
\For { $ C_i \in cache$ } 
       \State $Age(C_i) \gets Age(C_i)+1 $
\EndFor
\State $k \gets k+1 $
\EndLoop
\end{algorithmic}
\end{algorithm} 

\subsection{A Universal Upper Bound on the performance of caching policies}
We consider the ideal scenario where the cache size is infinite. Then, the duration that a content $i$ in the cache remains \textit{fresh} is limited only by its \textit{freshness} specification $F(i)$. The following theorem evaluates hit-rates under this scenario.

\begin{theorem}
For a system consisting of $n$ objects with infinite cache-size, and freshness specification $ F(i)$ for object $i$, under policy $\mathcal{A} \in \mathbb{A}$ where $\mathbb{A}$ is the set of all replacement policies,
\begin{align}
\hspace{-2.5mm} h_{\mathcal{A}}(i)=\frac{(F(i)-1)p_{i}}{1+(F(i)-1)p_{i}} \; ; \; i\in \lbrace1,2,\ldots,n \rbrace, \; \forall \; \mathcal{A} \in \mathbb{A}.  \label{hitInf}
\end{align}
\label{thmInf}
\end{theorem}
\vspace{-7mm}
To prove the above theorem, we use the fact that a fresh content $i$ exists in cache in the present slot \emph{if and only if} there is a single arrival of content $i$ into the cache in the past $F(i)-1$ slots. Refer to Section \ref{main_results} for a detailed proof. The next theorem shows that this infinite cache hit-rate is an upper bound for the hit-rates of any caching policy. 

\begin{theorem}
For a system consisting of $n$ objects with finite cache-size $m$, and freshness specification $ F(i)$ for object $i$, under policy $\mathcal{A} \in \mathbb{A}$ where $\mathbb{A}$ is the set of all replacement policies,
\begin{align}
h_{\mathcal{A}}(i) & \leq h^U(i) \; ; \; i\in \lbrace1,2,\ldots,n \rbrace, \; \; \forall \; \mathcal{A} \in \mathbb{A}, \; \; \text{where} \; \;  h^U(i) \triangleq \frac{(F(i)-1)p_{i}}{1+(F(i)-1)p_{i}} .
\label{hitUpperdefn}
\end{align}
\label{thmU}
\end{theorem}
\vspace{-7mm}

The proof for this theorem (refer to Section \ref{main_results}) uses the fact that the non-arrival of content $i$ into the cache in the past $F(i)-1$ slots implies that a fresh content $i$ does not exist in the cache in the present slot. This result leads us to a universal upper bound for the cache hit probability,
\begin{align}
\mathbb{P}_\mathcal{A}&(\text{hit}) \leq P^U, \; \; \forall \; \mathcal{A} \in \mathbb{A}, \; \; \text{where} \; \; P^U  \triangleq \sum_{i=1}^{n} p_{i} h^U(i)=  \sum_{i=1}^{n} p_{i} \frac{(F(i)-1) p_{i}}{1+(F(i)-1) p_{i}}. \label{Pupper}
\end{align}

\subsection{Least Popular (LP) Policy}
 The LP policy refers to a class of policies which endeavour to asymptotically estimate the popularity distribution of contents by observing the request stream, and then retain contents of higher popularity in the cache.
 
The formal definition of LP policy, (\emph{i.e.}, a policy which prioritizes contents based on their popularity) incorporated in Step \ref{Algo1step14} of Algorithm \ref{firstalgo} is as follows:

\begin{algorithm} 
\caption*{\textbf{LP Definition:}} \label{LPdefn}
\begin{algorithmic}[1]
        \If  { $p_{X_k} > p_i $ for some $ C_i \in \text{cache} $}
          \State Replace $C_{LP}$ with ${C_{X_k}}$; where 
          \Statex \hspace{4.5mm} \textit{LP} $ \triangleq \argmin_i \lbrace p_i: C_i \in \text{cache} \rbrace $
        \EndIf
\end{algorithmic}
\end{algorithm} 

The following theorem proves that the hit-rates for the $m$ most popular contents attain their \emph{upper bounds} whereas the remaining contents possess \emph{zero} probability of hit.
\vspace{-3mm}

\begin{theorem}
For a system with $n$ objects indexed in the decreasing order of popularity, cache-size $m$ ($m < n$), freshness specification $ F(i)$ for object $i$ and where $h^U(i)$ is defined as in \eqref{hitUpperdefn},
\begin{align}
h_{LP}(i) &= h^U(i) \; ; && i\in \lbrace1,2,\ldots,m\rbrace,  \nonumber \\
&= 0 \; ; && i\in \lbrace m+1,m+2\ldots,n \rbrace. \label{hitLP} 
\end{align}
\label{thmLP}
\end{theorem}

\vspace{-12mm}

The proof essentially follows from the fact that in the steady state, the cache contains only the $m$ most popular contents. Refer to Section \ref{main_results} for the proof. Next, we analyze the performance of the LP policy under the practical scenario in which the number of objects is ever increasing. 

\subsubsection*{Asymptotic analysis of hit probability} 
We characterize the hit probability $\mathbb{P}_{LP}(\text{hit})$ under a scaling of the system with respect to the number of objects and the cache-size. The setting for the analysis is as follows:
\begin{itemize}
\item[1.] For the sake of simplicity, we let $ F(i)=F \hspace{2mm}\forall \; i $ and refer to it as the \emph{freshness parameter}.
\item[2.] We scale the system by letting $n \rightarrow \infty$, and let $M(n)$ and $F(n)$ denote the cache-size and the freshness parameter scaling functions respectively.
\item[3.] Contents are indexed in the decreasing order of popularity and   follow the Zipf's law with parameter $\beta$. In order to enable scaling of the model, we assume $\beta > 1$ so that the Zipf distribution is well-defined as $n \rightarrow \infty$.

\end{itemize}


The following theorem shows that, under the above setting, as long as the cache size and the  freshness parameter are increasing with $n$, the LP policy is asymptotically optimal.

\vspace{-3mm}
\begin{theorem}
For a system consisting of $n$ objects with cache-size $M(n)$ and freshness parameter $F(n)$, if $ M(n)=\omega(1) $ and $ F(n)=\omega(1) $, then $ \underset{n \rightarrow \infty }{\lim} \mathbb{P}_{LP}(\normalfont \text{hit}) = 1$.
\label{asymptote}
\end{theorem}
Here, since $ \mathbb{P}_{LP}(\text{hit}) < P^U < 1 $, we have $ \mathbb{P}_{LP}(\text{hit}) \rightarrow P^U $, \emph{i.e.}, the LP policy attains upper bound performance as $n \rightarrow \infty$. Moreover, the cache performance asymptotically reaches full efficiency. Refer to Section \ref{main_results} for a detailed proof.

\subsection{Least Recently Used (LRU) Policy}
The LRU policy replaces the least recently requested content from the cache with the fetched content. The formal definition of the LRU policy incorporated in Step \ref{Algo1step14} of Algorithm \ref{firstalgo} is as follows:

\begin{algorithm}
\caption*{\textbf{LRU Definition:}} \label{LRUdefn}
\begin{algorithmic}[1]
        \State Replace $C_{LRU}$ with ${C_{X_k}}$; where \newline \textit{LRU} $ \triangleq \argmin_i \lbrace \text{last used time-slot of } C_i: C_i \in \text{cache} \rbrace $
\end{algorithmic}
\end{algorithm}

\subsubsection*{An Approximation for LRU Hit-rate}

\par We now derive an approximation to the hit-rate under the LRU policy.
Denote by $T_c(i)$, the number of time slots by which $m$ distinct contents other than $C_i$ are requested at least once. We refer to $T_c(i)$ as the \textit{characteristic time} of content $i$. 

\begin{figure}[H]
\centering
\begin{tikzpicture}[scale=0.5]
\draw[gray, thick, ->] (1,0) -- (11,0);
\draw[black, thick] (1,0.2) -- (1,-0.2);
\draw[black,thick] (1,0.7) node[anchor=south] {Content $`i$' requested};
\draw[black, thick] (2,0.2) -- (2,-0.2);
\draw[black, thick] (3,0.2) -- (3,-0.2);
\draw[black, thick] (4,0.2) -- (4,-0.2);
\draw[black, thick] (5,0.2) -- (5,-0.2);
\draw[black, thick] (6,0.2) -- (6,-0.2);
\draw[black, thick] (7,0.2) -- (7,-0.2);
\draw[black, thick] (8,0.2) -- (8,-0.2);
\draw[black, thick] (9,0.2) -- (9,-0.2);
\draw[black, thick] (10,0.2) -- (10,-0.2);
\draw[black, thick, ->] (1,0.9) -- (1,0.5);
\draw[decoration={brace,mirror,raise=5pt},decorate]
  (2,-0.5) -- node[below=6pt] {$T_c(i)$ slots} (10,-0.5);
\draw[black, thick, ->] (10,0.9) -- (10,0.5);
\draw[black,thick] (10,0.7) node[anchor=south] {Content $`i$' replaced};
\draw[black,thick] (1,-1.3) node[anchor=south] {\text{\small{$t=0$}}}; 
\end{tikzpicture}
\caption{It takes $T_c(i)$ time slots for content $i$ to become the least recently used amongst cached contents.} \label{fig:LRU}
\end{figure}

Suppose that $C_i$ is requested in the present slot ($t=0$). Assuming that the contents in the cache are ordered as most recently used first, $C_i$ currently occupies the top position. In the ensuing time, it takes $m$ distinct requests apart from $C_i$ for $C_i$ to shuffle to the bottom of the cache and get evicted. This is true provided $C_i$ is not requested again before getting evicted. Under this condition, $T_c(i)$ denotes the number of time slots for which $C_i$ stays in the cache before getting replaced. Hence, quantifying $ T_c(i) $ becomes important for calculating the hit-rate of $C_i$ under the LRU policy. The characteristic time $T_c(i)$ can be mapped to the celebrated \textit{Coupon Collector's Problem} described below. 

\noindent \textbf{Coupon Collector's Problem:} Given a collection $\mathcal{C}=\lbrace 1,2,\ldots,n \rbrace$ of \textit{n} coupons with $p_i$ being the probability of drawing coupon $i$, determine the number of independent draws (with replacement) from $\mathcal{C}$ to first obtain a collection of \textit{m} different coupons.

Let $T_m$ denote the number of draws referred to as the \emph{waiting time} for the Coupon Collector's Problem. The expected number of draws has been shown to satisfy the following equation in \cite{Flajolet}:
\begin{align}
\mathbb{E}(T_m) = \sum_{q=0}^{m-1}(-1)^{m-1-q}{n-q-1 \choose n-m } \sum_{|J|=q} \frac{1}{1-P_J}, \; \; \text{where} \; \; \displaystyle P_J=\sum_{i \in J} p_i. \label{exp1}
\end{align}
\noindent  

In the context of the LRU policy with a cache of size $m$, $T_c(i)$ is analogous to the waiting time $T_m$ for the Coupon Collector's Problem. In particular, the request for a content $i$ in the former corresponds to a coupon $i$ being drawn from the collection $\mathcal{C} $ in the latter. However, the only difference is that, unlike $T_m$, $T_c(i)$ considers the occurence of coupon $i$ as a blank draw. We now make approximations for calculating $T_c(i)$ which are similar in spirit to those of Che \emph{et al.} \cite{Che2002} and \cite{Fricker}.
\subsubsection{Calculation of characteristic time $T_c(i)$} 
Let $\overline{T}_c $ denote the number of time-slots until $m+1$ distinct contents are requested. Let $Z_1 $ denote the first request. Now, if the first request is for content $i$, \emph{i.e.}, $Z_1 = i $, then $ T_c(i) = \overline{T}_c - 1 $. Using this as the basis, we derive the following result in Section \ref{main_results}.
\begin{align}
\mathbb{E}(\overline{T}_c)=1+ \sum_{i=1}^{n} \, \mathbb{E} \left( T_c(i) \right) \, \mathbb{P} ( Z_1 = i ).  \label{Expapprox}
\end{align} 

\textit{Approximation 1:} The dependence of $\mathbb{E} (T_c(i))$ on $i$ can be ignored, \emph{i.e.}, $\mathbb{E} (T_c(i)) \approx t_c \; \forall \; i $.

This is a reasonable approximation when the individual popularities are relatively insignificant to their sum and becomes exact if the requests are equiprobable. We support this claim for Zipf distributed requests using numerical simulations provided in Section \ref{valid_approx}. We notice that the error in approximating $\mathbb{E} (T_c(i))$ with $ t_c $ reduces with decreasing values of $ \beta $, with increasing values of the $ \frac{m}{n}$ ratio and also as $ n \rightarrow \infty $ for a fixed $ \frac{m}{n}$ ratio. Using this approximation in \eqref{Expapprox}, we get 
\begin{align}
\mathbb{E} (T_c(i)) \approx t_c \triangleq \mathbb{E}(\overline{T}_c) - 1, \; \; \forall \; i \in \lbrace 1,2,\ldots,n\rbrace.  \label{tc_approx}
\end{align} 

\textit{Approximation 2:} We assume that for large $n$, the random variable $T_c(i) $ is well approximated by its expected value (\emph{i.e.}, it is nearly deterministic).

In the next Section, we provide analytical justifications for the above approximation under Zipf distributed requests with parameter $\beta \in [0,1)$. In addition, our numerical simulations (refer to Section \ref{valid_approx}) validate this approximation for several values of $ \beta$ (including $\beta > 1$) and under a wide range of scaling functions used for $m$ with respect to $n$. Using this approximation in equation \eqref{tc_approx}, we get $ T_c(i) \approx t_c, \; \; \forall \; i \in \lbrace 1,2,\ldots,n\rbrace. $ In the works by \cite{Che2002} and \cite{Fricker}, $t_c$ is referred to as the characteristic time of the \emph{cache}. Finally, from the definition of $\overline{T_c}$, it follows that $ \overline{T_c}  = T_{m+1} $. Hence, $\mathbb{E}( \overline{T_c} )$ can be calculated from equation \eqref{exp1}, which is then used in equation \eqref{tc_approx} to compute $ t_c $.


\subsubsection{Hit-rate for LRU policy}
 We consider hit-rates under the following two ideal scenarios:
 \begin{itemize}
 \item \textit{Scenario 1:} If the cache size is infinite, then the duration that a content $i$ remains \emph{fresh} in the cache is limited by its \textit{freshness} specification $F(i)$. Then, the hit-rate is given by Theorem \ref{thmInf} as $ \displaystyle h(i)= h^U(i),\; i\in \lbrace1,2,\ldots,n\rbrace. $
 \item \textit{Scenario 2:} If the freshness specification of content $i$, $F(i)$ is infinite, then the duration that it resides in the cache since its last request is limited by the \textit{characteristic time} which arises because of finite cache size. Let $ \tau_i$ denote the inter-arrival time between requests for content $i$ under the IRM; $ \tau_i \sim \textit{Geom}(p_i)$. For a content $i$ requested in the present slot, a request after $ \tau_i $ slots will result in a cache hit if and only if $ \tau_i < T_c(i) $. Therefore, the hit-rate is given by $ \displaystyle h(i)= \mathbb{P}(\tau_i < T_c(i)) \approx \mathbb{P}(\tau_i < t_c) =1-(1-p_i)^{ t_c -1} \; ,\; i\in \lbrace1,2,\ldots,n\rbrace. $ This hit-rate could be considered as the discrete time equivalent of Che's approximation given in \cite{Che2002},\cite{Fricker} where there are no freshness constraints.
 \end{itemize}

In our scenario, both cache size and the freshness are finite quantities. Hence, it is reasonable to approximate the LRU hit-rate by the minimum of the ideal hit-rates under the above two scenarios. Specifically, our hit-rate approximation for LRU is given by
\begin{align}
h_{LRU}(i)\approx \min \left( \frac{(F(i)-1)p_i}{1+(F(i)-1)p_i} \; ; \; 1-(1-p_i)^{ t_c -1}  \right). \label{LRUapprox}
\end{align}
\vspace{-5mm}

In Section \ref{simulation}, we illustrate using numerical simulations that the above approximation is generally quite accurate.

\section{Asymptotic analysis of characteristic time for LRU policy}
\label{no_name}
In this section, we provide analytical results to justify our approximation that $T_c(i)$ is nearly determinstic for large $n$. As mentioned earlier, the characteristic time $T_c(i)$ of the LRU policy is analogous to the waiting time $T_m$ for the Coupon Collector's Problem. Hence, in the rest of this section, we use the parlance of the Coupon Collector's Problem. Since the asymptotic behaviour of $T_m$ and $T_c(i)$ are identical, we characterize $T_m$ instead of $T_c(i)$ to simplify the analysis.

In the following, we provide a limit theorem for the case where the coupon draws obey the Zipf's law with Zipf parameter $\beta \in [0,1).$ Note that the Zipf's law generalises the uniform coupon draw case for which, the asymptotic distribution for the waiting time has been investigated in the works \cite{erdos}, \cite{baum}, \cite{anna}.


\begin{theorem}
Consider the coupon draws being sampled from a Zipf distribution with parameter $\beta \in [0,1)$, such that the $n$ coupons are indexed in the decreasing order of popularity. If $ m= o ( n^{\frac{1-\beta}{2-\beta}})$, then $ T_m \overset{i.p}{\rightarrow} m$.  
\label{baum_zipf}
\end{theorem}

In context of the LRU policy, the above theorem implies that if the cache-size $ m$ scales slower than $ n^{\frac{1-\beta}{2-\beta}}$, then asymptotically, $T_c(i)$ converges to $m$ in probability. Refer to Section \ref{main_results} for a detailed proof. The study by \cite{baum} provided limiting distributions for $T_m$ depending on how $m$ scales with $n$ as $ n \rightarrow \infty$. In particular, for the uniform case, \cite{baum} derived conditions on the scaling of $m $ under which $ T_m \overset{i.p}{\rightarrow} m$. This becomes a corollary to Theorem \ref{baum_zipf} for $\beta=0$ and is stated below.

\begin{cor}
Consider the coupon draws being equiprobable. If $m= o( \sqrt{n}) $, then $ T_m \overset{i.p}{\rightarrow} m$..  
\label{baum}
\end{cor}


Next, we justify Approximation 2 for the uniform coupon popularity. To this end, we derive concentration bounds for the deviations of $ T_m$ from its mean. 
\begin{theorem}
Consider the coupon draws being equiprobable. If $m= o(n) $, then for any $\delta > 0 $,
\begin{align*}
\mathbb{P} \left( \frac{T_m}{\mathbb{E}(T_m)} < (1 - \delta) \right) \leq \exp \left(- \frac{m}{2} \; \delta^2 \right) \; \; \text{and} \; \; \mathbb{P} \left( \frac{T_m}{\mathbb{E}(T_m)} > (1 + \delta) \right) \leq \exp \left( - \sqrt{\frac{n}{m}} \; \delta^{\frac{3}{2}} \right).
\end{align*}
\label{upper_lower_chernoff_tail}
\end{theorem}
\vspace{-8mm}
The above theorem provides upper and lower tail bounds for $ T_m$. Due to the assymmetric nature of the upper and lower tail bounds, depending on the scaling function of the cache-size, one of the tail bounds could become tight around the mean faster than the other as $ n \rightarrow \infty$. In particular, for $m = o (n^{\frac{1}{3}})$, the upper tail bound becomes tight faster than the lower tail bound. The vice-versa holds true for $m = \omega (n^{\frac{1}{3}})$. Interestingly, a cache-size scaling of $m = \Theta (n^{\frac{1}{3}})$ provides the best gaurantee for a tight concentration of $T_c(i)$ about its mean.


\begin{theorem}
Consider the coupons draw being equiprobable. If $m= \omega(1) $ and $m= o(n) $, then $ T_m \overset{i.p}{\rightarrow}  \mathbb{E}(T_m) $.  
\label{chernoff}
\end{theorem}

In comparison with Corollary \ref{baum}, the above theorem shows that even under relaxed constraints for the scaling of $m$, convergence of $T_m$ to its mean in probability is still achieved. Theorems \ref{upper_lower_chernoff_tail} and \ref{chernoff} imply that for the case of equiprobable requests, Approximation 2 is justified for any sub-linear scaling of cache-size which is a reasonable assumption in practical scenarios. On the other hand, for a more general distribution of requests \big(Zipf, $\beta \in [0,1)$\big), Approximation 2 is valid under the constraints on cache-size scaling imposed by Theorem \ref{baum_zipf}. 
Deriving concentration bounds similar to Theorem \ref{upper_lower_chernoff_tail} for general values of $\beta$ appears more challenging and is left for future work. However, we provide empirical evidence which validates Approximation 2 for several values of $ \beta$ and a wide range of cache-size scalings.



\subsection*{Empirical Validation of approximations for LRU policy}  \label{valid_approx}
\subsubsection{Simulations to validate Approximation 1}
We first obtain $\mathbb{E} (T_c(i)) $ for all $ i \in \lbrace1,2,\ldots,n \rbrace$ and $ \mathbb{E}(\overline{T}_c)$ by averaging these quantities over sufficiently large number of simulations. We then use $ \mathbb{E}(\overline{T}_c)$ and obtain $t_c$ from equation \eqref{tc_approx} to calculate the error in approximating $\mathbb{E} (T_c(i)) $. This error denoted by $ \mu $ is given by $ \displaystyle \mu = \frac{|\mathbb{E} (T_c(i)) - t_c|}{\mathbb{E} (T_c(i))}  $.

\begin{figure*}
\begin{minipage}[b]{0.5\linewidth}
\centering 
\includegraphics[width=1.05\linewidth]{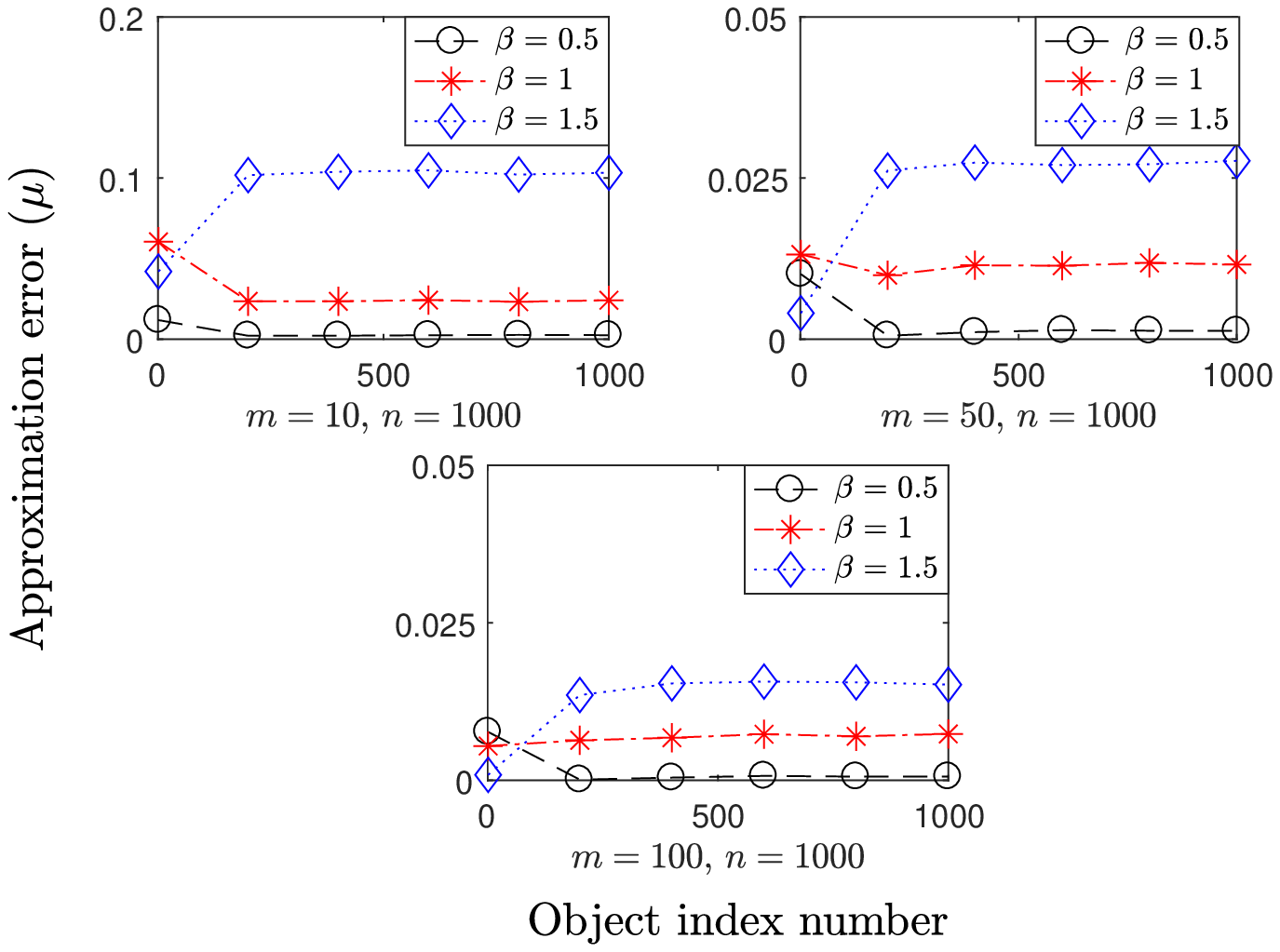}
\caption*{(a)}
\end{minipage} 
\begin{minipage}[b]{0.5\linewidth}
\centering
\includegraphics[width=1.05\linewidth]{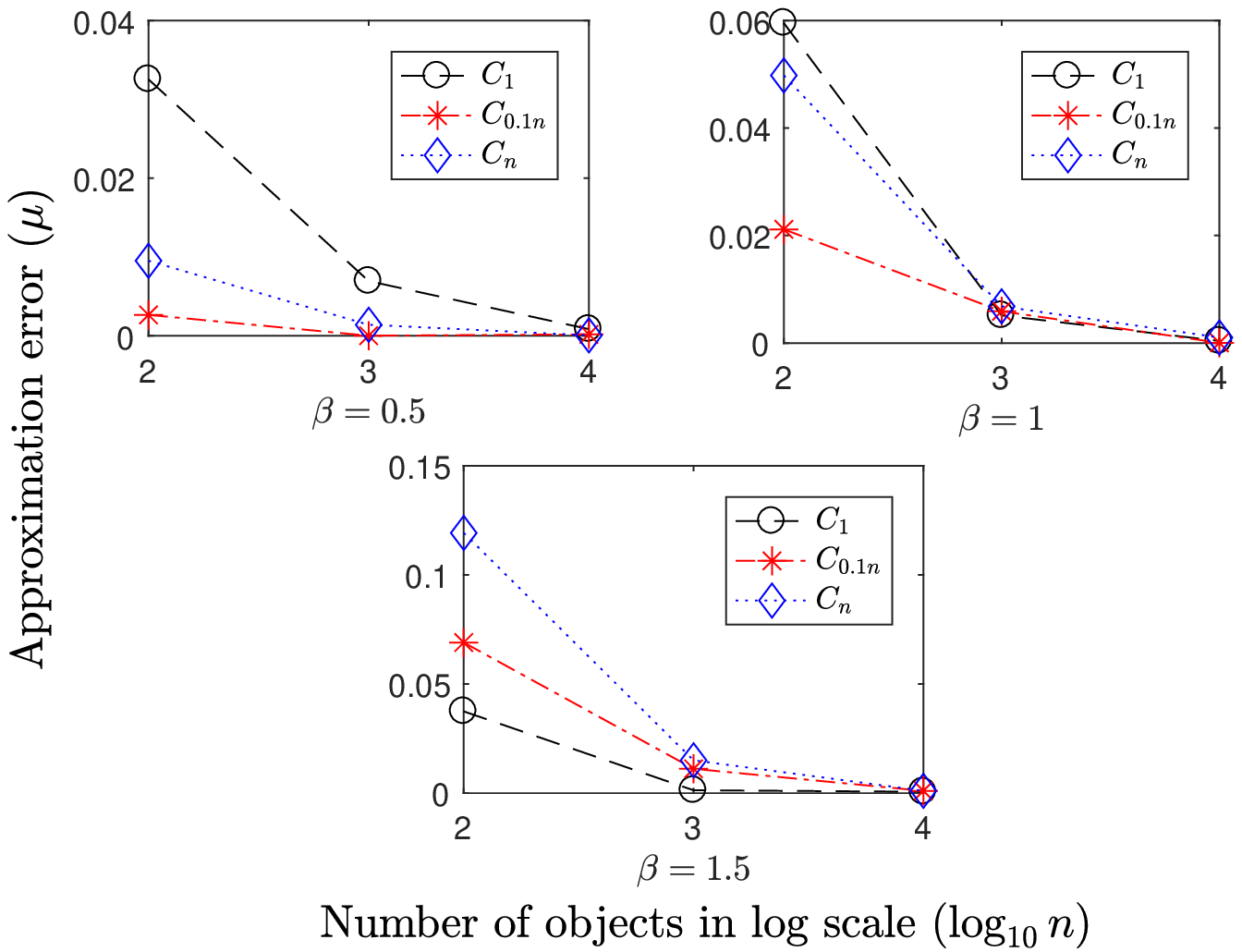}
\caption*{(b)}
\end{minipage}
\caption{(a) Error in approximating $\mathbb{E} (T_c(i)) $ against the object index $i$ for values of Zipf parameter $\beta=0.5,1,1.5$ and (b) asymptotic trend for the error in approximating $\mathbb{E} (T_c(i)) $ as the number of objects, $n \rightarrow \infty $. Here, cache-size $m=0.1n$ and $\beta=0.5,1,1.5$.}
\label{mu}
\end{figure*}  

In Figure \ref{mu}$(a)$, we plot the error against content index for a total number of contents $n=1000$. We consider different $ \frac{m}{n}$ ratios and different values of $ \beta $ to obtain the plots. We notice that the error in approximating $\mathbb{E} (T_c(i))$ with $ t_c $ reduces with decreasing values of $ \beta $ and increasing values of the $ \frac{m}{n}$ ratio. Further, in Figure \ref{mu}$(b)$, we plot the error against the number of objects $n$ for contents with indices $1, 0.1n$ and $n$ denoted by $ C_1, C_{0.1n}$ and $ C_n$ respectively. We fix values for $ \beta$ and $\frac{m}{n}$ and observe that the error decays to zero as $ n \rightarrow \infty $. From this, we infer that Approximation 1 becomes better with an increase in the number of objects.

\subsubsection{Simulations to validate Approximation 2}
To show that $ T_c(i) $ concentrates around $\mathbb{E}( T_c(i) ) $ as $ n \rightarrow \infty $, it is sufficient to show that its coefficient of variation $ \displaystyle c_v(T_c(i)) \triangleq \frac{\sigma( T_c(i) )}{\mathbb{E}( T_c(i) )} \rightarrow 0 $ as $ n \rightarrow \infty $. Here $\sigma( T_c(i) )$ denotes the standard deviation for $ T_c(i) $. We plot $c_v( T_c(i) ) $ against $ n $ in Figures \ref{Tc_i_approx}$(a)$ and \ref{Tc_i_approx}$(b)$ for contents with indices $1, 0.1n$ and $n$ denoted by $ C_1, C_{0.1n}$ and $ C_n$ respectively. We perform simulations for different values of $ \beta $ and different scaling functions for the cache-size with respect to $n$ -- namely, \emph{linear}, \emph{square root} and \emph{logarithmic} scaling.  

\begin{figure*}
          \begin{minipage}[b]{0.5\linewidth}
\centering
\includegraphics[width=1.05\linewidth]{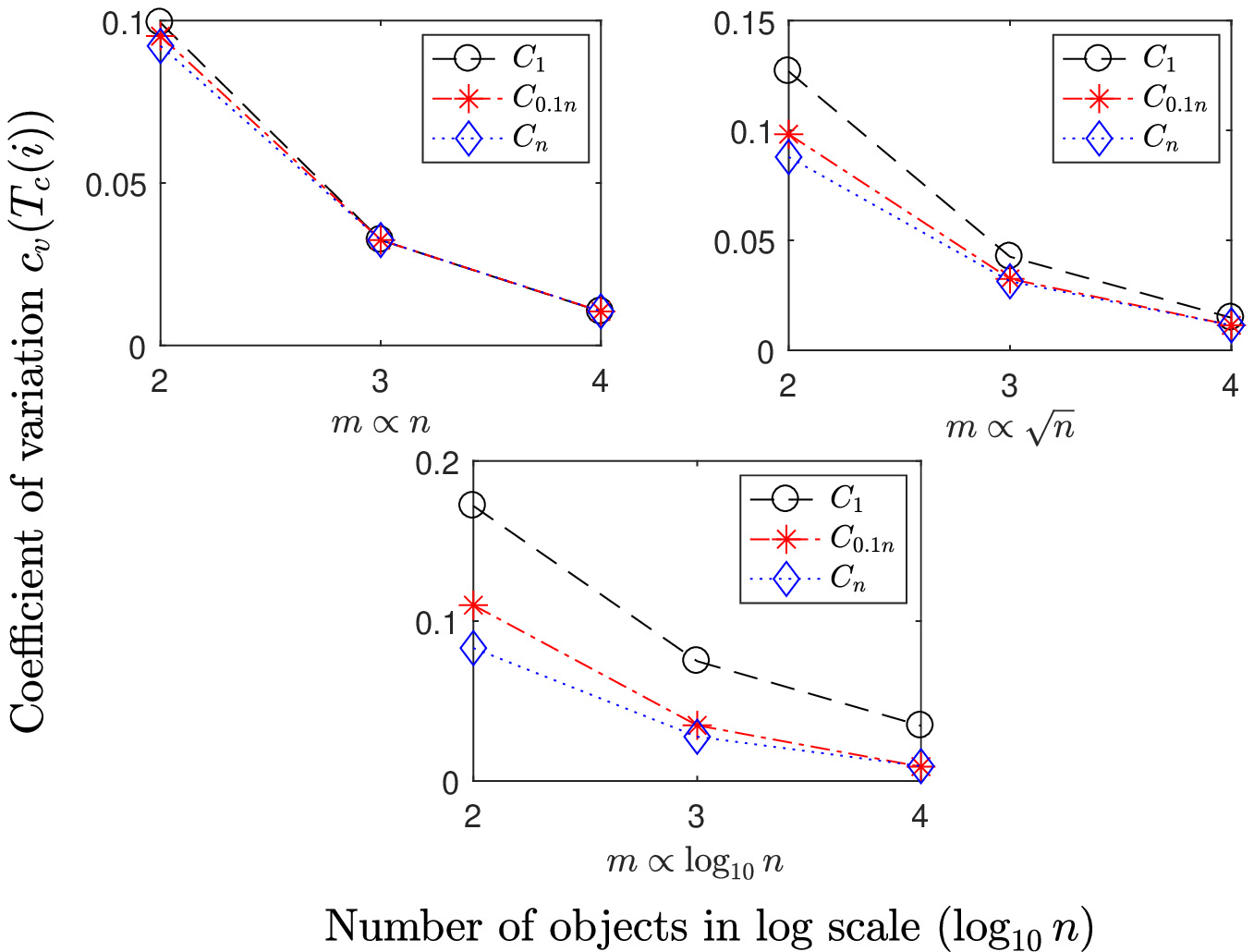}
\caption*{(a) $\beta = 0.5$}
\end{minipage} 
\begin{minipage}[b]{0.5\linewidth}
\centering
\includegraphics[width=1.05\linewidth]{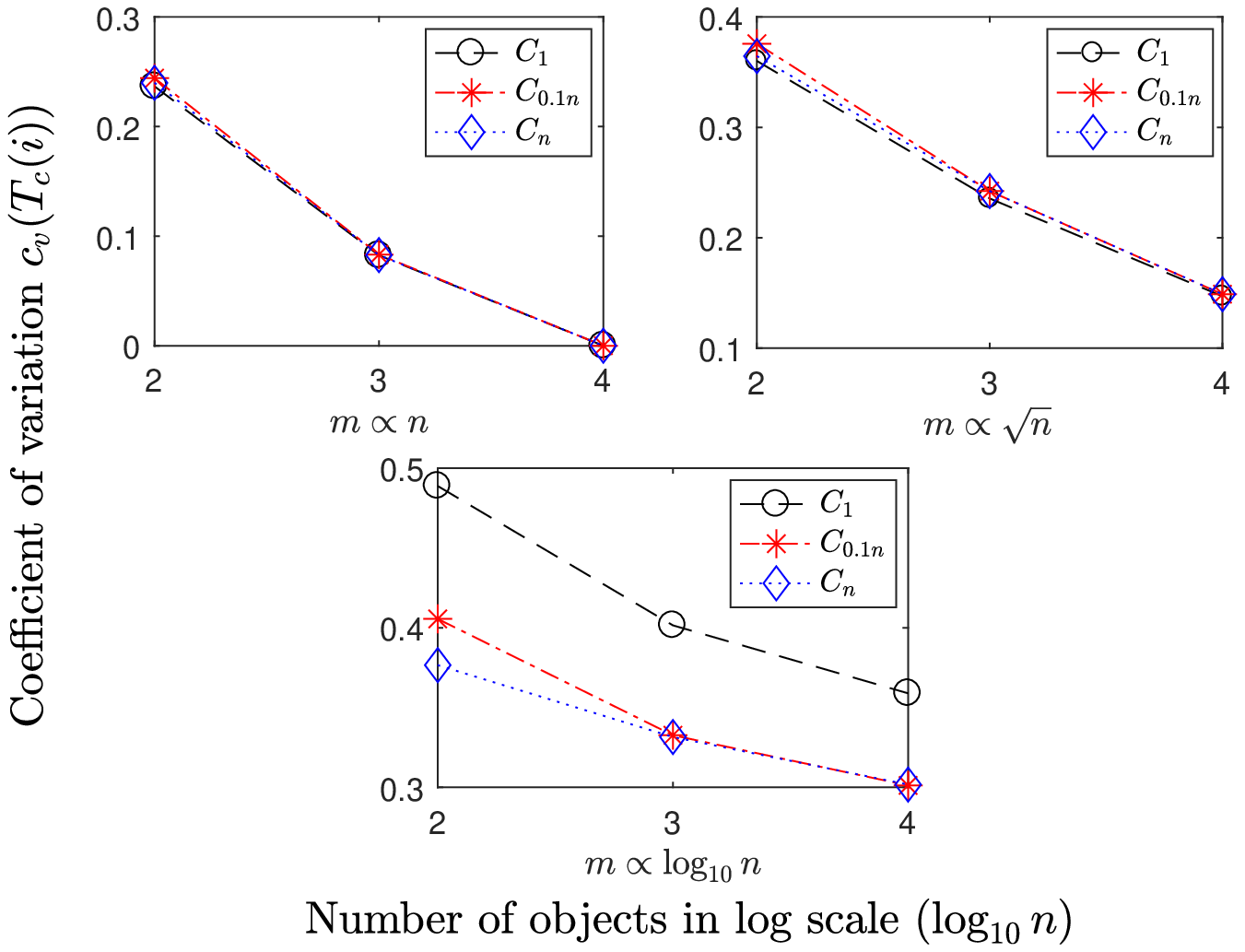}
\caption*{(b) $\beta = 1.5 $.}
\end{minipage}
\caption{Asymptotic trend for the coefficient of variation of $ T_c(i) $ as the number of objects, $n \rightarrow \infty $.}
\label{Tc_i_approx}
\end{figure*}

We observe that $ c_v( T_c(i) ) $ has a decaying trend towards zero as $ n \rightarrow \infty $. This supports our claim in Approximation 2. Further, we observe that $ c_v( T_c(i) ) $ decays faster for lower values of $ \beta $. Also, we observe a faster decay in $ c_v( T_c(i) )$ as we move from a \emph{logarithmic} scaling of cache-size  to the \emph{square root} and \emph{linear} scalings.

\section{Policies which consider freshness of contents} \label{aware}
In this section, we assume that the server knows the freshness specification $F(i)$ for all objects $ i\in \lbrace1,2,\ldots,n \rbrace $. We study the class of policies $ \mathbb{M} $, which explicitly use this freshness information in their cache replacement strategy. Firstly, we state the most basic feature that any freshness aware policy must incorporate -- namely, using the knowledge of $F(i)$ to discard stale data from the cache. More specifically, a content $i$ must be allowed to reside in the cache only for $F(i)$ time slots since its arrival, after which it must be evicted. With this feature, the general template for a policy $\mathcal{M} \in \mathbb{M}$ is given in Algorithm \ref{secondalgo}. Using this template, we modify the traditional LP and LRU policies. We incorporate the LP and LRU definitions in Step \ref{Algo2step14} of Algorithm \ref{secondalgo} to obtain the Modified-LP (M-LP) and Modified-LRU (M-LRU) policies respectively. Our next result shows that the M-LP policy achieves better performance than the traditional LP policy.

\begin{algorithm}
\caption{Freshness-aware caching.} \label{secondalgo}
\begin{algorithmic}[1]
\State $k \gets 1$.
\Loop
\For { $ C_i \in cache$ }
   \If {$Age(C_i) > F(i) $}
       \State Remove $C_i$ from cache \Comment{stale data}
   \EndIf
\EndFor
\If { $C_{X_k} \in \text{cache}$ }
    \State serve the request \Comment{cache hit}
\Else    
    \State Fetch content type $C_{X_k}$ from back-end
    \State serve the request  \Comment{cache miss}
    \If {cache is full}
        \State Implement cache replacement policy $ \mathcal{M} $ \label{Algo2step14}
        \State {\textbf{if} $C_{X_k} \in \text{cache}$ \textbf{then} $Age(C_{X_k}) \gets 1$ }
    \Else   
         \State Place $C_{X_k}$ in cache. $Age(C_{X_k}) \gets 1$
    \EndIf     
\EndIf     
\For { $ C_i \in cache$ } 
       \State $Age(C_i) \gets Age(C_i)+1 $
\EndFor
\State $k \gets k+1 $
\EndLoop
\end{algorithmic}
\end{algorithm}

\vspace{-5mm}
\begin{theorem}
For a system with $n$ objects indexed in the decreasing order of popularity, cache-size $m$ ($m < n$), freshness specification $ F(i)$ for object $i$ and $h^U(i)$ defined as in \eqref{hitUpperdefn},
\begin{align}
& h_{M\text{-}LP}(i) = h^U(i) \; ; \hspace{8.5mm} i\in \lbrace1,2,\ldots,m\rbrace.  \nonumber \\
& 0 < h_{M\text{-}LP}(i) \leq h^U(i) \; ; \hspace{2mm} i\in \lbrace m+1,m+2,\ldots,n \rbrace.  \label{hitMLP} \\
\intertext{} \nonumber
\end{align}
\label{thmMLP}
\end{theorem}

\vspace{-33mm}
This follows from arguments similar to those made for proving Theorem \ref{thmLP} except that the contents $\lbrace m+1,m+2,\ldots,n \rbrace $ possess positive hit probabilities. This is because, the continuous eviction of the top $m$ popular contents as and when they get stale allows contents $\lbrace m+1,m+2,\ldots,n \rbrace $ to enter the cache. From Theorems \ref{thmLP} and \ref{thmMLP} we have, $ h_{LP}(i)  < h_{M\text{-}LP}(i) ; \, i\in \lbrace m+1,m+2,\ldots,n \rbrace.$ Therefore $ \mathbb{P}_{LP}(\text{hit})  < \mathbb{P}_{M\text{-}LP}(\text{hit}) $. Hence, M-LP outperforms LP. This shows that caching redundancies in the form of stale contents is strictly sub-optimal.

\subsection*{Least Expected Hits (LEH) Policy}

We now propose a new policy called the LEH policy which takes into account both popularity and freshness specifications of contents in order to make content replacement decisions. As the name suggests, under the LEH policy, content priority is decided by determining the expected number of hits that it would generate, if it were to be retained (or introduced in the cache, as the case may be). In contrast to M-LP and M-LRU policies, which only use freshness specifications to evict stale data, LEH policy additionally uses them to prioritize among the cached contents. 

In the present time slot, consider a content $i$ which has resided in the cache for $T$ slots since its arrival $(T \leq F(i))$. It can be retained in the cache for a maximum duration of $ F(i)-T $ further slots. Let $N_i$ denote the number of possible future hits for content $i$ if it were not discarded. Then, $ N_i \sim \text{Binomial}(F(i)-T, p_i),$ where $ \mathbb{E}(N_i) = \left( F(i)-T \right) p_i. $ On the other hand, consider a content $j$ fetched from the back-end. If it were to be placed in the cache in the present slot, then the number of possible hits it would have generated if it were retained for the future $ F(j)-1 $ slots would be $ N_j \sim \text{Binomial}(F(j)-1, p_j), $ where $\mathbb{E}(N_j)  = \left( F(j)-1 \right) p_j.$ 

Now, consider the event that a content arrives and the cache is full. The LEH policy essentially makes its content replacement decision by comparing the expected number of future hits of the arriving content, with the expected number of future hits of the cached contents.

A formal definition of the LEH policy incorporated in Step \ref{Algo2step14} of Algorithm \ref{secondalgo} is as follows:

\begin{algorithm}
\caption*{\textbf{LEH Definition:}} \label{LEHdefn}
\begin{algorithmic}[1]
        \If  {$ \mathbb{E} ( N_{X_k} ) > \mathbb{E} ( N_{i} )$ for some $ C_i \in \text{cache} $}
          \State Replace $C_{LEH}$ with ${C_{X_k}}$; where 
          \Statex \hspace{4.5mm} \textit{LEH} $\triangleq \argmin_i \lbrace  \mathbb{E} ( N_{i} ) : C_i \in \text{cache} \rbrace $.
        \EndIf
\end{algorithmic}
\end{algorithm}

In Section \ref{simulation}, we illustrate through simulations that this framework of prioritizing contents indeed serves to improve the caching performance.

\section{Proof of main results} \label{main_results}
\begin{figure}[H]
\centering
\begin{tikzpicture}[scale=0.5]
\draw[gray, thick, ->] (1,0) -- (11,0);
\draw[black, thick] (1,0.2) -- (1,-0.2);
\draw[black, thick] (2,0.2) -- (2,-0.2);
\draw[black, thick] (3,0.2) -- (3,-0.2);
\draw[black, thick] (4,0.2) -- (4,-0.2);
\draw[black, thick] (5,0.2) -- (5,-0.2);
\draw[black, thick] (6,0.2) -- (6,-0.2);
\draw[black, thick] (7,0.2) -- (7,-0.2);
\draw[black, thick] (8,0.2) -- (8,-0.2);
\draw[black, thick] (9,0.2) -- (9,-0.2);
\draw[black, thick] (10,0.2) -- (10,-0.2);
\draw[decoration={brace,mirror,raise=5pt},decorate]
  (1,-0.5) -- node[below=6pt] {$(F(i)-1)$ slots} (9,-0.5);
\draw[black, thick, ->] (10,0.9) -- (10,0.5);
\draw[black,thick] (10,0.7) node[anchor=south] {$t=0$};
 
\end{tikzpicture}
\caption{Typically, a content $i$ in the cache stays \emph{fresh} for $F(i)-1$ slots after its arrival unless it is replaced prematurely.} \label{append}
\end{figure}

\subsection{Proof of Theorem \ref{thmInf}}
\begin{proof}
We consider a discrete time slot $t=0$ in the steady state (refer Figure \ref{append}) referred here as the present slot and observe the previous $F(i)-1$ slots indexed by $ t \in \lbrace -(F(i)-1), \ldots, -1 \rbrace $. The following three events $ \lbrace X_t,H_t,M_t \rbrace $ could have occurred in slot $t$. 
\begin{itemize}
\item[$a.$] $X_t$ = \{No request for content $i$\}; $\mathbb{P}(X_t)= 1-p_i $.
\item[$b.$] $H_t$ = \{Cache hit for content $i$\}; $\mathbb{P}(H_t)= p_i h_\mathcal{A}(i) $.
\item[$c.$] $M_t$ = \{Cache miss for content $i$\}; $\mathbb{P}(M_t)= p_i \left( 1- h_\mathcal{A}(i) \right) $. This corresponds to an arrival of content $i$ in the cache.
\end{itemize}

Owing to infinite cache-size, any content $i$ entering the cache remains \emph{fresh} for the entire $F(i)$ slots and the cache dynamics is indifferent to any policy $ \mathcal{A} \in \mathbb{A}$.  This leads us to state the following.

\begin{prop}
Fresh content $i$ exists in cache in the present slot if and only if there is a single arrival for $i$ in the past $F(i)-1$ slots.
\label{prop1}
\end{prop}

Let $A$ denote the event that there is a single arrival for $i$ in the past $F(i)-1$ slots. This is analogous to the event that there is a single cache miss for content $i$ in the past $(F(i)-1)$ slots.
From Proposition \ref{prop1}, the hit-rate can be obtained as follows:
\begin{align*}
h_{\mathcal{A}}(i) &= \mathbb{P}(A) = \mathbb{P} \left( \bigcup_{t=-(F-1)}^{1} M_t \right) \overset{(a)}{=} \sum_{t=-(F-1)}^{1} \mathbb{P}( M_t ) = (F(i)-1)p_i \left( 1- h_{\mathcal{A}}(i) \right). 
\end{align*}
Step $(a)$ follows since $ \big\{ M_t : t \in \lbrace -(F(i)-1), \ldots, -1 \rbrace \big\} $ is a set of mutually disjoint events.
\begin{align*}
\therefore h_{\mathcal{A}}(i) &= \frac{(F(i)-1)p_{i}}{1+(F(i)-1)p_{i}} \; , \; && i\in \lbrace1,2,\ldots,n \rbrace.
\end{align*}

\vspace{-3mm}
\end{proof}

\subsection{Proof of Theorem \ref{thmU}}
\begin{proof}
We consider a discrete time slot $t=0$ in the steady state (refer Figure \ref{append}) referred here as the present slot and observe the previous $F(i)-1$ slots indexed by $ t \in \lbrace -(F(i)-1), \ldots, -1 \rbrace $. The following three events $ \lbrace X_t,H_t,M_t \rbrace $ could have occurred in slot $t$. 
\begin{itemize}
\item[$a.$] $X_t$ = \{No request for content $i$\}; $\mathbb{P}(X_t)= 1-p_i $.
\item[$b.$] $H_t$ = \{Cache hit for content $i$\}; $\mathbb{P}(H_t)= p_i h_\mathcal{A}(i) $.
\item[$c.$] $M_t$ = \{Cache miss for content $i$\}; $\mathbb{P}(M_t)=$ $p_i \left( 1- h_\mathcal{A}(i) \right) $ where $M_t$ comprises of two disjoint events:\\
\indent $M_t^A$ = \{Cache miss for content $i$ resulting in an arrival of content $i$ in the cache.\}, 

\noindent $M_t^{NA}$ = \{Cache miss for content $i$ not resulting in an arrival of content $i$ in the cache\}.

\end{itemize}

\noindent Then, clearly the following proposition about the present slot holds irrespective of the caching policy $ \mathcal{A} \in \mathbb{A}$ being used. 

\begin{prop}
No arrival of content $i$ in the past $F(i)-1$ slots implies that a fresh content $i$ does not exist in the cache in the present slot.
\label{prop2}
\end{prop}

Let $B$ denote the event that there is no arrival of content $i$ in the past $F(i)-1$ slots.
Using Proposition \ref{prop2} we have,
\begin{align}
1- h_\mathcal{A}(i) & \geq \mathbb{P} ( B ) = 1- \mathbb{P} ( B^c ) = 1- \mathbb{P} \left( \bigcup_{t=-(F-1)}^{1} M_t^A \right) \overset{(a)}{\geq}  1- \mathbb{P} \left( \bigcup_{t=-(F-1)}^{1} M_t \right). \nonumber
\end{align}
Step $(a)$ holds as $  M_t^A \subseteq M_t $.
\begin{align}
\therefore h_\mathcal{A}(i) & \leq \mathbb{P} \left( \bigcup_{t=-(F-1)}^{1} M_t \right) \overset{(b)}{\leq}  \sum_{t=-(F-1)}^{1} \mathbb{P} ( M_t ) = (F(i)-1)p_i \left( 1- h_{\mathcal{A}}(i) \right). \nonumber 
\end{align}
Step $(b)$ follows by applying the Union Bound.
\begin{align}
\therefore h_{\mathcal{A}}(i) & \leq \frac{(F(i)-1)p_{i}}{1+(F(i)-1)p_{i}} = h^U (i), & i\in \lbrace1,2,\ldots,n \rbrace. \nonumber 
\end{align}
\vspace{-5mm}
\end{proof}

\subsection{Proof of Theorem \ref{thmLP}}
\begin{proof}
A content $i \in \lbrace1,2,\ldots,m \rbrace$ entering the cache remains in it for the entire $F(i)$ slots and cannot be replaced by any other content during this time. This is due to the fact that it belongs to the set of $m$ most popular objects. As a result, Proposition \ref{prop1} follows and hence the hit-rates derived in Theorem \ref{thmInf} hold for the $m$ most popular objects.
\begin{align*}
\therefore h_{LP}(i) &= \frac{(F(i)-1)p_{i}}{1+(F(i)-1)p_{i}} \; ; && i\in \lbrace1,2,\ldots,n \rbrace.
\end{align*}
Also, once a content $i \in \lbrace1,2,\ldots,m \rbrace$ enters the cache, beyond $F(i)$ time slots, it is only replaced by its own fresh copy on request. This ensures that in the steady state, no content $i \in \lbrace m+1,2,\ldots,n \rbrace$ enters the cache. Therefore $ h_{LP}(i) = 0, \quad i\in \lbrace m+1,\ldots,n \rbrace.$
\end{proof}
\vspace{-5mm}
\subsection{Proof of Theorem \ref{asymptote}}
\textsc{Proof Outline}: 
\begin{itemize}
\item[1.] To show that $ \underset{n \rightarrow \infty }{\lim} \mathbb{P}_{LP}(hit) = 1$, we obtain a lower bound to $ \mathbb{P}_{LP}(hit) $ and show that it converges to $1.$ We do this by considering the system under two cases.
\item[2.]  Firstly, we show that under the scaling $ M(n)= O \left( F(n)^{\frac{1}{\beta}}\right) $, $ \mathbb{P}_{LP}(hit) = 1-O \left( \frac{1}{M(n)^{\beta-1}} \right)$. Now given $ M(n)=\omega(1) $, it follows that  $ \underset{n \rightarrow \infty }{\lim} \mathbb{P}_{LP}(hit) = 1$.
\item[3.] Secondly, we show that under the scaling $ M(n)= \omega \left( F(n)^{\frac{1}{\beta}}\right) $, $ \mathbb{P}_{LP}(hit) = 1-O \left( \frac{1}{F(n)^{\frac{\beta-1}{\beta}}} \right)$. Now, given $ F(n)=\omega(1) $, it follows that  $ \underset{n \rightarrow \infty }{\lim} \mathbb{P}_{LP}(hit) = 1$.
\end{itemize}

\begin{proof}
We use the following upper bound on Zipf distribution in our derivation: \newline
 $\displaystyle \sum_{i=1}^{n} i^{-\beta} > \int_{1}^{n+1} i^{-\beta} di \geq \frac{0.9}{\beta-1} $ for large enough n. Hence, $\displaystyle  p_{i} \leq \frac{(\beta-1)}{0.9} i^{-\beta} \quad \forall \; i $. Starting the derivation, we have
\vspace{-5mm}
\begin{align}
\mathbb{P}_{LP}(hit) = \sum_{i=1}^{M(n)} \frac{(F(n)-1) p_{i}^2}{ 1+(F(n)-1) p_i } &= \sum_{i=1}^{M(n)} \frac{(F(n)-1) p_{i}^2 + p_i - p_i}{ 1+(F(n)-1) p_i } \nonumber \\ 
&= \sum_{i=1}^{M(n)} p_i - \sum_{i=1}^{M(n)} \frac{p_i}{1+ \left(  F(n)-1 \right)p_i }. \label{PLP}
\end{align}
\vspace{-2mm}
Let $ \displaystyle A \triangleq \sum_{i=1}^{M(n)} p_i = 1-\sum_{i=M(n)+1}^{n} p_i $, then replacing $p_i$ with its upper bound yields
\vspace{-3mm}
\begin{align}
A & \geq 1- \sum_{M(n)+1}^{n} \frac{\beta-1}{0.9} i^{-\beta} = 1- \frac{\beta-1}{0.9} \int_{M(n)}^{n} i^{-\beta} \mathrm{d}i = 1- \frac{1}{0.9} \left( \frac{1}{\left( M(n) \right)^{\beta-1}} -\frac{1}{n^{\beta-1}} \right). \label{A}
\end{align}
\vspace{-2mm}
Let $ \displaystyle S \triangleq \sum_{i=1}^{M(n)} \frac{ p_i }{1+(F(n)-1) p_i } $, then replacing $p_i$ with its upper bound yields 
\begin{align}
S & \leq \sum_{i=1}^{M(n)} \frac{1}{ \frac{0.9}{\beta-1} i^\beta + \left(  F(n)-1 \right) }. \label{S}
\end{align}

\noindent Now, given that $ M(n)=\omega(1) $ and $ F(n)=\omega(1) $, we analyze the following two cases:

\textsc{Case $(i)$:} Let $ M(n)= O \left( F(n)^{\frac{1}{\beta}}\right) $. In this case, the second term in the denominator of $S$ dominates. Hence $S$ can be upper bounded as follows.
\begin{align}
S & \leq \sum_{i=1}^{M(n)} \frac{1}{ F(n)-1 } = \frac{M(n)}{F(n)-1}. \label{S1} 
\end{align}
Then from equations \eqref{PLP}, \eqref{A} and \eqref{S1} we have
\begin{align}
\mathbb{P}_{LP}(hit) &= A - S \geq 1- \frac{1}{0.9} \left( \frac{M(n)}{\left( M(n) \right)^{\beta}} -\frac{1}{n^{\beta-1}} \right) - \frac{M(n)}{F(n)-1}. \nonumber 
\end{align}
Here, since $ M(n)= O \left( F(n)^{\frac{1}{\beta}}\right) $, as $n \rightarrow \infty$ the term $\frac{M(n)}{\left( M(n) \right)^{\beta}}$ dominates $\frac{M(n)}{F(n)-1}$.

\vspace{2mm}
\noindent $ \displaystyle \therefore \mathbb{P}_{LP}(hit) = 1-O \left( \frac{1}{M(n)^{\beta-1}} \right).$ Now, given $ M(n)=\omega(1) $, it follows that  $ \underset{n \rightarrow \infty }{\lim} \mathbb{P}_{LP}(hit) = 1$. 

\vspace{2mm}
\textsc{Case $(ii)$:} Now, let $ M(n)= \omega \left( F(n)^{\frac{1}{\beta}}\right) $, then $ M(n) > F(n)^{\frac{1}{\beta}} $ for large enough $n$. Using this to split the limits of inequality \eqref{S}, we obtain
\begin{align}
S &\leq \sum_{i=1}^{F(n)^{\frac{1}{\beta}}} \frac{1}{ \frac{0.9}{\beta-1} i^\beta + \left(  F(n)-1 \right) } + \sum_{F(n)^{\frac{1}{\beta}}}^{M(n)}  \frac{1}{ \frac{0.9}{\beta-1} i^\beta + \left(  F(n)-1 \right) }. 
\end{align}
Here since $ M(n)= \omega \left( F(n)^{\frac{1}{\beta}}\right) $, as $ n \rightarrow \infty$, $F(n)-1$ dominates the denominator of the first term whereas $ \frac{0.9}{\beta-1} i^\beta$ dominates the denominator of the second term. Hence, we have 
\begin{align}
S  \leq \sum_{i=1}^{F(n)^{\frac{1}{\beta}}} \frac{1}{ F(n)-1 } + \sum_{F(n)^{\frac{1}{\beta}}}^{M(n)}  \frac{ \beta-1 }{0.9} i^{-\beta}
&= \frac{F(n)^{\frac{1}{\beta}}}{ F(n)-1 } + \frac{ \beta-1 }{0.9} \int_{F(n)^{\frac{1}{\beta}}-1}^{M(n)} i^{-\beta} \mathrm{d}i, \nonumber \\
& = \frac{F(n)^{\frac{1}{\beta}}}{ F(n)-1 } + \frac{1}{0.9}
\Bigg[ \frac{1}{ \left( F(n)^{\frac{1}{\beta}} -1 \right)^{\beta-1} } - \frac{1}{M(n)^{\beta-1}} \Bigg]. \label{S2} 
\end{align}
Then from equations \eqref{PLP}, \eqref{A}, \eqref{S2} and putting $F(n)-1 \approx F(n)$ for large $n$  we have
\begin{align}
\mathbb{P}_{LP}(hit) &\geq 1 + c_1 \frac{1}{n^{\beta-1}} - c_2 \frac{1}{ F(n)^{\frac{\beta-1}{\beta}} }, \quad \text{where} \;\; c_1= \frac{1}{0.9}  \;\; \text{and} \;\; c_2= 1+\frac{1}{0.9}. 
\label{S3}
\end{align}

\vspace{-2mm}
\noindent $ \displaystyle \therefore \mathbb{P}_{LP}(hit) = 1-O \left( \frac{1}{F(n)^{\frac{\beta-1}{\beta}}} \right).$ Now, given $ F(n)=\omega(1) $, it follows that  $ \underset{n \rightarrow \infty }{\lim} \mathbb{P}_{LP}(hit) = 1$.

\end{proof}

\subsection{Proof for Equation \eqref{Expapprox} }
\begin{proof}
Let $\overline{T}_c $ denote the number of time-slots until $m+1$ distinct contents are requested. Let $Z_1 $ denote the first request. Now, if the first request is for content $i$, \emph{i.e}, $Z_1 = i $, then $ T_c(i) = \overline{T}_c - 1 $. More formally, for $ a \in \mathbb{Z} $;
\begin{align}
\mathbb{P}(\overline{T}_c-1 > a \, \vert \, Z_1=i) &= \mathbb{P}( T_c(i) > a). \label{Result1}
\end{align}
Also, from Baye's theorem on total probability,
\begin{align}
\mathbb{P}(\overline{T}_c-1 > a ) = \sum_{i=1}^{n} \mathbb{P}(\overline{T}_c-1 > a \, \vert \, Z_1=i) \; \mathbb{P}(Z_1=i).  \nonumber
\end{align}
Substituting equation \eqref{Result1} in the above equation we have,
\begin{align*}
\mathbb{P}(\overline{T}_c-1 > a ) = \sum_{i=1}^{n} \mathbb{P}( T_c(i) > a) \; \mathbb{P}(Z_1=i). 
\end{align*}
Taking summation on both sides over $a$ and interchanging limits in the RHS gives,
\begin{align}
\sum_{a=-1}^{\infty} \mathbb{P}(\overline{T}_c > a+1 ) = \sum_{i=1}^{n} \left( \sum_{a=-1}^{\infty} \mathbb{P}( T_c(i) > a) \right) \mathbb{P}(Z_1=i). \label{Result2}
\end{align} 
Next, we simplify the following two terms from the above equation.
\begin{align*}
\sum_{a=-1}^{\infty} \mathbb{P}( T_c(i) > a) &= 1+ \sum_{a=0}^{\infty} \mathbb{P}( T_c(i) > a) = 1+ \mathbb{E}(T_c(i)), \\
\hspace{-13mm} \text{and} \; \sum_{a=-1}^{\infty} \mathbb{P}(\overline{T}_c > a+1 ) &= \sum_{b=0}^{\infty} \mathbb{P}(\overline{T}_c > b ) = \mathbb{E}( \overline{T}_c ),
\end{align*}  
where $ b=a+1 $. Therefore, equation \eqref{Result2} becomes,
\begin{align*}
\mathbb{E}( \overline{T}_c ) &= \sum_{i=1}^{n} \Big( 1+ \mathbb{E} \big( T_c(i) \big) \Big) \; \mathbb{P}(Z_1=i) = 1+ \sum_{i=1}^{n} \mathbb{E} \big( T_c(i) \big) \; \mathbb{P}(Z_1=i).
\end{align*}

\vspace{-3mm}
\end{proof}

\subsection{Proof of Theorem \ref{baum_zipf}}
\textsc{Proof Outline}: 
\begin{itemize}
\item[1.] We express the waiting time as $ T_m= \sum_{k=1}^{m} X_k $ where $ X_k$ denotes the number of draws to obtain the $k^{th}$ distinct coupon having obtained a set of $k-1$ distinct coupons. Note that under the assumption of independent coupon draws, $ \lbrace X_1, X_2, \ldots X_m  \rbrace $ are independent.
\item[2.] To prove that $ T_m \overset{i.p}{\rightarrow} m$, it is sufficient to show that $ \mathbb{P}(T_m = m) \rightarrow 1 .$ To do this, we obtain a lower bound to $ \mathbb{P}(T_m = m)$ and show that it converges to $1.$
\item[3.] To get this lower bound, we use the fact that $ T_m = m$ only if $ X_k = 1 $ for all $k.$ We then find a lower bound for $ \mathbb{P}( X_k =1)$ which results in a lower bound for $ \mathbb{P}(T_m = m).$  
\end{itemize}

\begin{proof}
Let $Z_i$ denote the $i^{th}$ distinct coupon obtained. Consider the event $ I \triangleq \lbrace Z_1= i_1, Z_2= i_2, \ldots, Z_{k-1}=i_{k-1} \rbrace$. Then $ X_k \, \vert \, I \sim \text{Geom} ( 1- \sum_{l=1}^{k-1} p_{i_l}) $. Hence, $ \mathbb{P}( X_k=1 \, \vert \, I) = 1- \sum_{l=1}^{k-1} p_{i_l} $. The unconditional probability can be given as
\begin{align*}
\mathbb{P}( X_k=1) &= \sum_{ \lbrace i_1, \ldots, i_{k-1} \rbrace = \lbrace 1,\ldots,n \rbrace } \mathbb{P}( X_k=1 \, \vert \, I) \mathbb{P}(I).
\end{align*}
where the summation runs through all $(i_1, \ldots, i_{k-1} )$ such that there exists an injection $ g : \lbrace 1,\ldots, k-1 \rbrace \rightarrow \lbrace 1, \ldots, n \rbrace $ satisfying $g(l) = i_l $ for $ l \in \lbrace 1,\ldots, k-1 \rbrace $. This summation can be lower bounded by the following inequality.
\begin{align}
\hspace{-3mm} \mathbb{P}( X_k=1) & \geq \inf_{ \lbrace i_1, \ldots, i_{k-1} \rbrace = \lbrace 1,\ldots,n \rbrace } \mathbb{P}( X_k=1 \, \vert \, I) = \inf_{ \lbrace i_1, \ldots, i_{k-1} \rbrace = \lbrace 1,\ldots,n \rbrace } 1- \sum_{l=1}^{k-1} p_{i_l}  \overset{(a)}{=} 1 - \sum_{j=1}^{k-1} p_j. \label{Xk_lower}
\end{align}
Step $(a)$ follows from the fact that the coupons are indexed in the  decreasing order of their popularities. This naturally leads to a lower bound for the waiting time as follows.
\begin{align}
\mathbb{P}(T_m = m) &= \mathbb{P}( X_k=1 ; \; \forall \; k \in \lbrace 1,2,\ldots,m \rbrace ) \nonumber \\
& \overset{(b)}{=} \prod_{k=1}^{m} \mathbb{P}( X_k=1 ) \overset{(c)}{\geq} \prod_{k=1}^{m} \left( 1 - \sum_{j=1}^{k-1} p_j  \right) \geq  \left( 1 - \sum_{j=1}^{m-1} p_j  \right)^{m-1}. \label{Tm_lower}
\end{align}
Step $(b)$ holds because $ \lbrace X_1, X_2, \ldots X_m  \rbrace $ are independent. Step $(c)$ follows from  \eqref{Xk_lower}. We now use a Lemma used in the subsequent portion of the proof.
\begin{lemma} 
For the Zipf distribution with Zipf parameter $\beta$ given by $ p_i \propto i^{-\beta} $ defined over the support $ \lbrace 1,2,\ldots,n\rbrace $, for $ m < n$ and $ \beta < 1$, 
\begin{align*}
\sum_{i=1}^{m} p_i = \theta \left( \frac{m}{n} \right)^{1-\beta}.
\end{align*}
\label{lemma1}
\end{lemma}



\begin{proof}
We begin with showing that for $\beta < 1$,  $\sum_{i=1}^{n} i^{-\beta} = \theta \left( \frac{n^{1-\beta} }{1- \beta} \right) .$

$ \displaystyle \sum_{i=1}^{n} i^{-\beta} \geq \int_{1}^{n+1} i^{-\beta} \mathrm{d}i = \frac{(n+1)^{1-\beta}-1}{1-\beta} \approx \frac{n^{1- \beta}}{1-\beta}$ for large $n$.

$ \displaystyle \sum_{i=1}^{n} i^{-\beta} \leq 1+ \int_{1}^{n} i^{-\beta} \mathrm{d}i = 1+ \frac{n^{1-\beta}-1}{1-\beta} = \frac{n^{1-\beta}-\beta}{1- \beta} \leq \frac{n^{1-\beta} }{1- \beta}.$ Hence, $ \displaystyle \sum_{i=1}^{n} i^{-\beta} = \frac{n^{1-\beta} }{1- \beta} $ for large $ n .$ This gives us $ \displaystyle \sum_{i=1}^{m} p_i = \frac{\sum_{i=1}^{m} i^{-\beta}}{ \sum_{i=1}^{n} i^{-\beta}} = \left( \frac{m}{n} \right)^{1-\beta} .$

\end{proof}

Using Lemma \ref{lemma1} in \eqref{Tm_lower} yields $ \displaystyle \mathbb{P}(T_m = m) \geq \left( 1- \left( \frac{m}{n}\right)^{1-\beta} \right)^m $. In this inequality, let the RHS be denoted by $A$, then $ \displaystyle \log A = m \log \left( 1- \left( \frac{m}{n}\right)^{1-\beta} \right) $. Now using $ \displaystyle \log x \geq 1- \frac{1}{x} $ we lower bound the \emph{log} term and simplify to obtain $ \displaystyle 0 \geq \log A \geq -\frac{m^{2- \beta}}{n^{1- \beta}} $. Now, given that $ \displaystyle m= o ( n^{\frac{1-\beta}{2-\beta}})$; $ \log A \rightarrow 0$ and hence $A \rightarrow 1$. As a result, $ \mathbb{P}(T_m = m) \rightarrow 1 .$
\end{proof}

\subsection{Proof of Theorem \ref{upper_lower_chernoff_tail}}
\textsc{Proof Outline}: 
\begin{itemize}
\item[1.] We have $ T_m= \sum_{k=1}^{m} X_i $ as before (refer to the proof of Theorem \ref{baum_zipf}). To obtain the lower tail bound, we apply Chernoff's bound on $T_m$ for a $\delta-$deviation from its mean which gives $\mathbb{P} ( T_m < \mathbb{E}(T_m)(1- \delta)) \leq \exp \big( \underset{ \lambda > 0}{\inf} f(\lambda) \big) $. Here, $f(\lambda) = d \lambda^2 - c \lambda$ with constants $c, d > 0$, which is minimized by $ \lambda_{\min}= \frac{c}{2d}$. This yields the required bound on further simplification.
\item[2.] Similarly, we obtain an upper tail bound as $\mathbb{P} ( T_m > \mathbb{E}(T_m)(1+ \delta)) \leq \exp \big( \underset{ \lambda > 0}{\inf} g(\lambda) \big) $. Here, $g(\lambda) = -a \lambda + ( e^\lambda-1) b $ with constants $a, b > 0$, which is minimized by $ \lambda_{\min}= \log\frac{a}{b}$. This yields the required bound on further simplification.
\end{itemize}

\begin{proof}
Here, for the uniform Coupon Collector's Problem, $ X_k \sim \text{Geom} (p_k) $ where $ p_k=\frac{n-(k-1)}{n}; k \in \lbrace 1,2, \ldots m\rbrace $. Let $ \mathbb{E}( X_k)$ be denoted as $ \mu_k = \frac{1}{p_k}$. It follows that $\mathbb{E}(T_m) = \sum_{k=1}^{m} \frac{n}{n-(k-1)}$ $ = n ( H_n - H_{n-m} )$ where $H_n$ denotes the $n^{th}$ harmonic number\footnote{The $n^{th}$ harmonic number is the sum of the reciprocals of the first $n$ natural numbers.}. It is known that $ H_n = \log n + \gamma +o(1) $ where $\gamma$ is known as the  \emph{Euler-Mascheroni} constant. So, for large $n$ we have $ \mathbb{E}(T_m) = n \log (\frac{n}{n-m})$. We now proceed to derive the tail bounds for the deviations of $ T_m$ about its mean.

\textsc{$(a)$ Lower tail bound:} Let $ S \triangleq T_m - \mathbb{E}(T_m)= \sum_{i=1}^{m} ( X_k - \mu_k ) $. Then for $ c > 0 $ and $ \lambda > 0$,
\begin{align*}
\mathbb{P} ( S < -c ) &= \mathbb{P}( e^{- \lambda S} \geq e^{ \lambda c}) \overset{(i)}{\leq} \frac{\mathbb{E}(e^{- \lambda S} )}{e^{ \lambda c}}  = e^{ - \lambda c} \prod_{k=1}^{m}\mathbb{E} ( e^{-\lambda(X_k - \mu_k)}).
\end{align*}

Step $(i)$ follows by using Markov's Inequality. Further, substituting $ \mu_k = \frac{1}{p_k}$ and $ \mathbb{E}(e^{-\lambda X_k}) = \frac{p_k e^{-\lambda}}{1-(1-p_k)e^{-\lambda}} $ in the above expression and rearranging the terms gives,
\begin{align*}
\mathbb{P} ( S < -c ) &= e^{-\lambda c} \prod_{k=1}^{m} \frac{e^{\frac{\lambda}{p_k}}}{1+ \frac{1}{p_k}( e^\lambda -1 )}, \\
& \overset{(ii)}{\leq} e^{-\lambda c} \prod_{k=1}^{m} \frac{e^{\frac{\lambda}{p_k}}}{1+\frac{\lambda}{p_k}} \overset{(iii)}{\leq} e^{-\lambda c} \prod_{k=1}^{m} e^{\frac{\lambda^2}{2 p_k^2}} = \exp \left( -\lambda c + \frac{1}{2} \lambda^2 \sum_{k=1}^{m} \frac{1}{p_k^2} \right) = e^{f(\lambda)}.
\end{align*}

Step $(ii)$ holds because $ e^\lambda-1 \geq \lambda, \;\; \forall \; \lambda \in \mathbb{R} $. Step $(iii)$ holds as $ \displaystyle \frac{e^x}{1+x} \leq e^{\frac{x^2}{2}}, \;\; \forall \; x \geq 0 $. Now here, $f(\lambda) = d \lambda^2 - c \lambda$ with constants $d = \frac{1}{2} \sum_{k=1}^{m} \frac{1}{p_k^2}$. Minimizing $f(\lambda)$ over $ \lambda > 0$; we obtain $ \displaystyle \mathbb{P} ( S < -c ) \leq \text{exp}( - \frac{c^2}{2 \sum_{k=1}^{m} \frac{1}{p_k^2} } ) $. We now substitute $ c= \mathbb{E}(T_m)\delta = ( n \log \frac{n}{n-m}) \delta$ for $ \delta > 0$ in this inequality to obtain
\begin{align}
\mathbb{P} \big( T_m < \mathbb{E}(T_m)(1- \delta) \big) \leq \text{exp} \left( - \frac{\log^2(\frac{n}{n-m} ) }{2 \sum_{k=1}^{m} \frac{1}{\left( n-(k-1) \right)^2 }} \: \delta^2 \right). \label{Tmeq}
\end{align}
Here, $ \displaystyle \sum_{k=1}^{m} \frac{1}{\left( n-(k-1) \right)^2 } \leq \int_{n-m}^{n} \frac{1}{k^2} \; \mathrm{d}k = \frac{m}{n(n-m)} $. Substituting this in \eqref{Tmeq} we get,
\begin{align*}
\frac{\log^2(\frac{n}{n-m} ) }{ \sum_{k=1}^{m} \frac{1}{\left( n-(k-1) \right)^2 }} &\geq \frac{\log^2(\frac{n}{n-m} ) }{\frac{m}{n(n-m)}} \overset{(iv)}{\geq} \left( 1+ \frac{(n-m)^2}{n^2} - 2 \frac{(n-m)}{n} \right) \frac{n(n-m)}{m} = m (1-\frac{m}{n}) \overset{(v)}{\approx} m .
\end{align*}

Step $(iv)$ follows by using the inequality $\log x \geq 1 - \frac{1}{x}$ and Step $(v)$ holds for large enough $n$ as $ \frac{m}{n} \rightarrow 0$. Using this lower bound in the inequality in \eqref{Tmeq} we get,
\begin{align}
\mathbb{P} \big( T_m < \mathbb{E}(T_m)(1- \delta) \big) & \leq \exp \left( - \frac{m}{2} \delta^2 \right) . \nonumber
\end{align}

\textsc{$(b)$ Upper tail bound:} For all $ \lambda > 0$,
\begin{align*}
\mathbb{P}(T_m > t) &= \mathbb{P}( e^{ \lambda T_m} \geq e^{ \lambda t}) \overset{(i)}{\leq} \frac{\mathbb{E}(e^{ \lambda T_m} )}{e^{ \lambda t}} = e^{ - \lambda t} \prod_{k=1}^{m} \frac{(1-\frac{k-1}{n}) e^\lambda}{ 1 - \frac{k-1}{n} e^\lambda }.
\end{align*}
Step $(i)$ follows by using Markov's Inequality. Now as $ m \ll n$, $ \frac{k-1}{n} \ll 1$. Hence we can use the approximation that $1-x \approx e^{-x} $ for $|x| \ll 1$.

\begin{align*}
\therefore \mathbb{P}(T_m > t) & \leq e^{ - \lambda t} \prod_{k=1}^{m} \frac{e^{-\frac{k-1}{n}}e^\lambda}{e^{-\frac{k-1}{n}e^\lambda}} = \exp \left( -\lambda t +m \lambda + ( e^\lambda-1) \sum_{k=1}^{m} \frac{k-1}{n}  \right).
\end{align*}
Now, substituting $t= \mathbb{E}(T_m)(1+ \delta) = ( n \log \frac{n}{n-m})(1+\delta)$ where $\delta > 0$ and $ \sum_{k=1}^{m} \frac{k-1}{n} = \frac{m(m-1)}{2n}$, we obtain
\begin{align}
\mathbb{P} \big( T_m > \mathbb{E}(T_m)(1+ \delta) \big) & \leq \text{exp} \big( -a \lambda + ( e^\lambda-1) b \big) = e^{ g(\lambda)}. \label{minimize} 
\end{align}
where $a= ( n \log \frac{n}{n-m})(1+\delta)-m$ and $ b = \frac{m(m-1)}{2n} .$ 
We now deduce certain properties of constants $a$ and $b$. We have $ 1 - \frac{m}{n} < e^{-\frac{m}{n}} \Rightarrow  n \log (\frac{n}{n-m}) >m \Rightarrow a > 0$. Now using $ \log x \geq 1 - \frac{1}{x}$ we have $a \geq n(1-\frac{n-m}{n})(1+\delta) - m = m \delta$ and for large $m, \; b \approx \frac{m^2}{2n} $. So $\frac{a}{b} \geq \frac{2n}{m} \delta $ which implies that $ \frac{a}{b} \rightarrow \infty$ as $ n \rightarrow \infty$. Next, minimizing $ g(\lambda) $ in \eqref{minimize} over $ \lambda > 0$; we obtain
\begin{align}
\mathbb{P} \big( T_m > \mathbb{E}(T_m)(1+ \delta) \big) & \leq \exp \Bigg( a \log \bigg( \frac{b}{a} \bigg) + a - b \Bigg) . \label{something}
\end{align}
We upper bound the log term using the inequality $ \log x \leq \frac{x-1}{\sqrt{x}}$ and simplify \eqref{something} to yield the below inequality. Here, Step $(ii)$ is justified for large $n$ as $ \frac{b}{a} \rightarrow 0.$
\begin{align*}
\mathbb{P} \big( T_m > \mathbb{E}(T_m)(1+ \delta) \big) & \leq \exp \Bigg( a- \left( 1-\frac{b}{a} \right) a \sqrt{\frac{a}{b}}  \Bigg) \overset{(ii)}{\approx} \exp \bigg(-a  \sqrt{\frac{a}{b}} \: \bigg). 
\end{align*}
\begin{align*}
\therefore \mathbb{P} \big( T_m > \mathbb{E}(T_m)(1+ \delta) \big) & \leq \exp \left( - \sqrt{\frac{n}{m}} \; \delta^{\frac{3}{2}} \right). 
\end{align*}
\vspace{-3mm}
\end{proof}

\section{Simulation Results} \label{simulation}
We simulate the caching process assuming that the content popularities obey Zipf's law with parameter $\beta$. The caching process is governed by the system parameters: number of objects $(n)$, cache-size $(m)$, Zipf parameter $(\beta)$, freshness specification $F(i)$ for content $i \in \lbrace 1,2,\ldots,n \rbrace $ and the policy being implemented. We use two freshness specification profiles: \emph{uniform} ($F(i)=F, \; \forall \; i$) and \emph{linear} ($F(i) \propto i $) where index $i$ corresponds to the $i^{th}$ most popular content. The linear profile is motivated by practical scenarios in which the contents requested often (popular) are required to be very \emph{recent} and hence tend to have smaller freshness specifications.


Firstly, we compare the performance of different policies studied in this paper with the upper bound on the performance of all caching policies given by \eqref{Pupper}. Figures \ref{varybetafull} and \ref{varybetafull2} plot the steady-state hit probability against varying $\beta$ for the \emph{uniform} and \emph{linear} freshness profiles respectively. Figures \ref{varybetafull}$(b)$ and \ref{varybetafull2}$(b)$ provide zoomed in plots of Figures \ref{varybetafull}$(a)$ and \ref{varybetafull2}$(a)$ respectively for values of $\beta$ from $0$ to $0.5$. We infer that for low values of $\beta$, LEH outperforms all other policies. Here, we have not plotted LP and LRU policies since it is evident from Figures \ref{varybetafull}$(a)$ and \ref{varybetafull2}$(a)$ that they are outperformed by the M-LP and M-LRU policies respectively.  

Figure \ref{varyF} compares caching performance under \emph{uniform} freshness profile by varying parameter $F$. We observe that LEH performs better than M-LP and M-LRU policies over moderate values of $F$ and otherwise is comparable to them. 

Secondly, we depict the accuracy of the LRU approximation in Figures \ref{hits_varybeta_varyF}$(a)$ and \ref{hits_varybeta_varyF}$(b)$ which correspond to hit-rate variations with respect to $\beta$ and $F$ respectively for selected content indices: 1, 10 and 100. Hit-rates are obtained by simulating the caching process for sufficiently long runs to ensure their high accuracy. We infer that the theoretical approximations match the simulations reasonably well. We notice that the M-LRU hit-rates also adhere closely to the LRU approximation. This suggests that the LRU performance is unaffected by cache redundancies and is comparable to M-LRU performance for small cache-size and uniform freshness demands.

\begin{figure*}[t!]
\begin{minipage}[t]{0.5\linewidth}
\centering
\includegraphics[width=\linewidth]{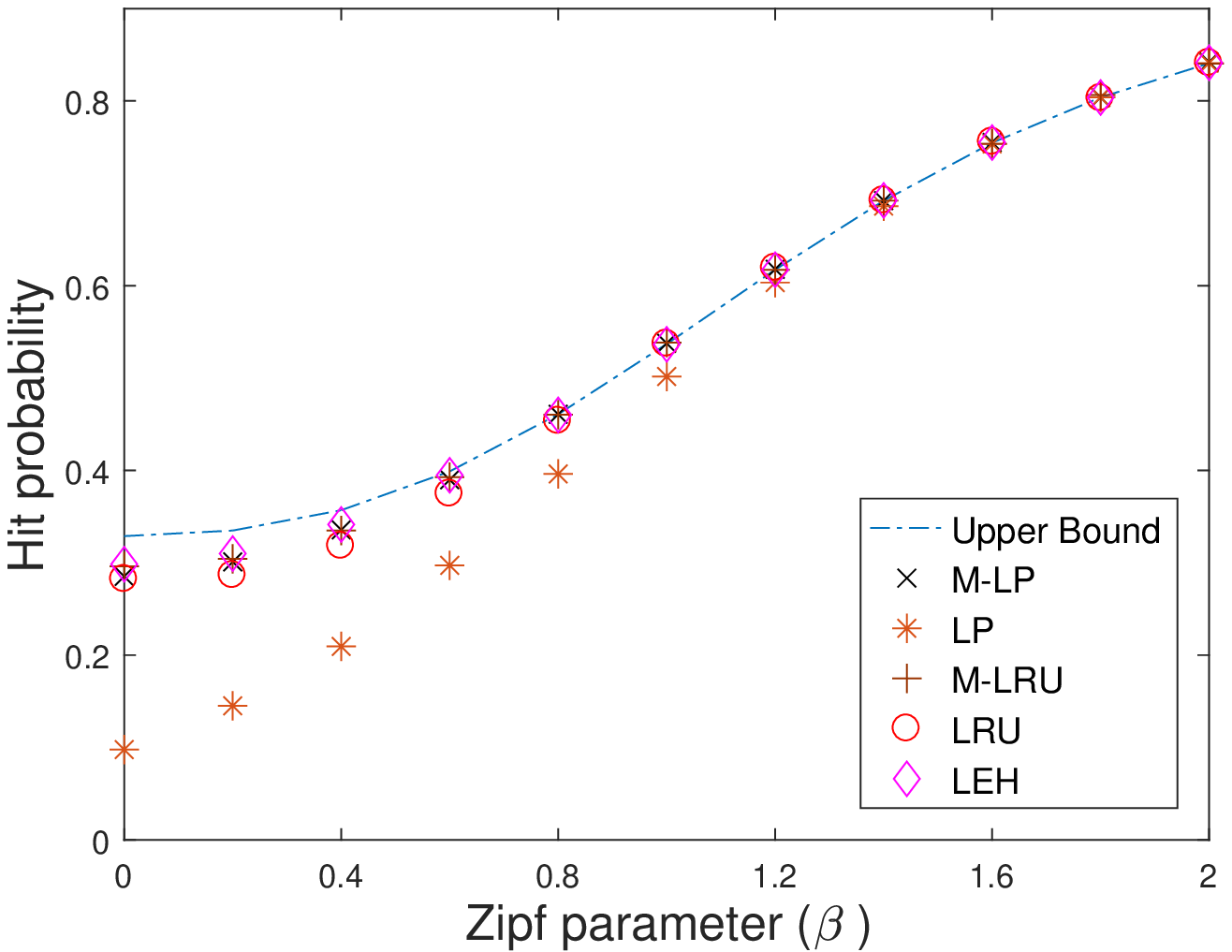}
 \caption*{(a)}
\end{minipage}
\hfill
\begin{minipage}[t]{0.5\linewidth}
\centering
\includegraphics[width=\linewidth]{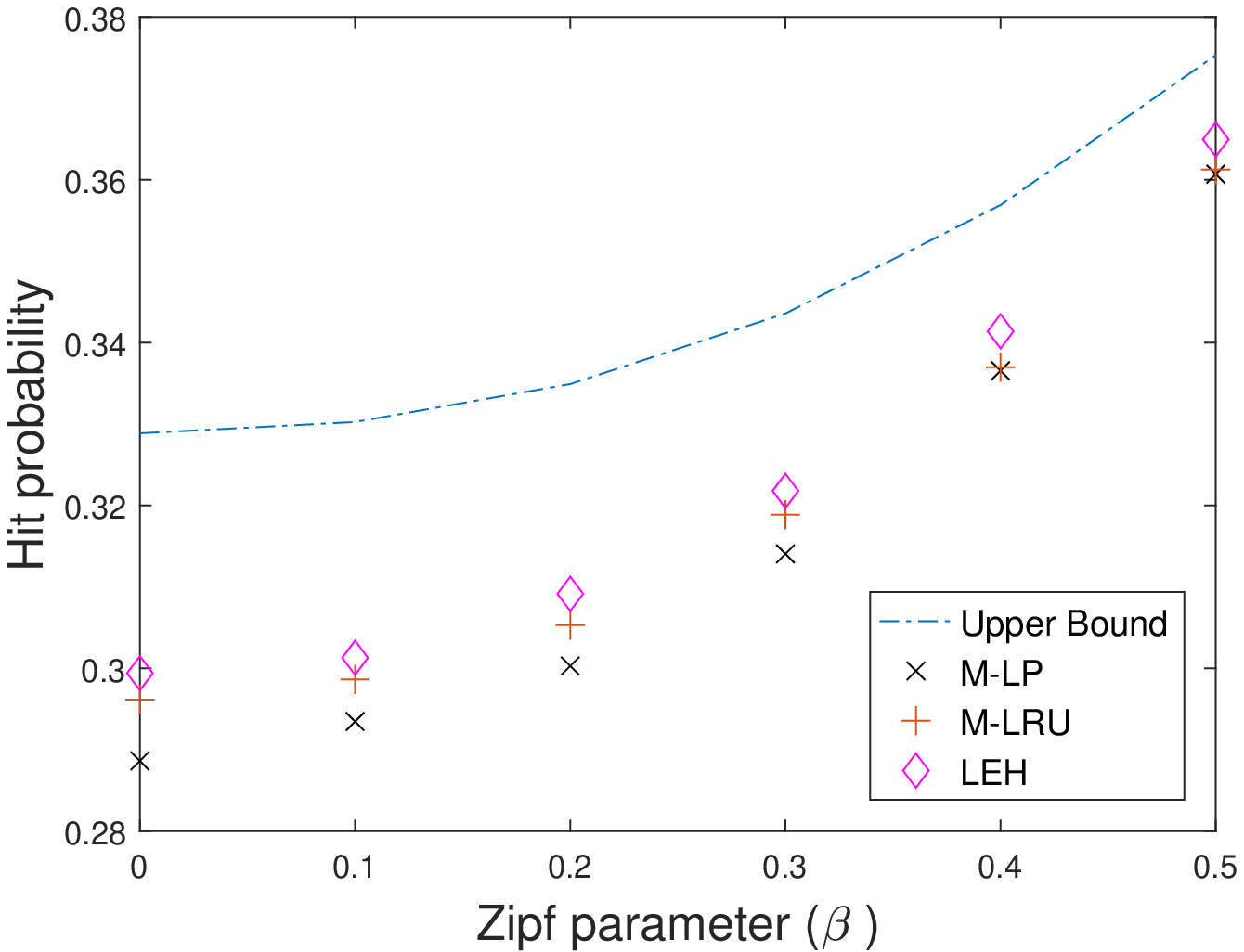}
 \caption*{(b)}
\end{minipage}
\caption{Performance of caching policies as a function of Zipf parameter $\beta$ for $ n=100$, $m=30 $ and $F=50$.}
\label{varybetafull} 
\end{figure*}

\begin{figure*}[t!]
\begin{minipage}[t]{0.5\linewidth}
\centering
\includegraphics[width=\linewidth]{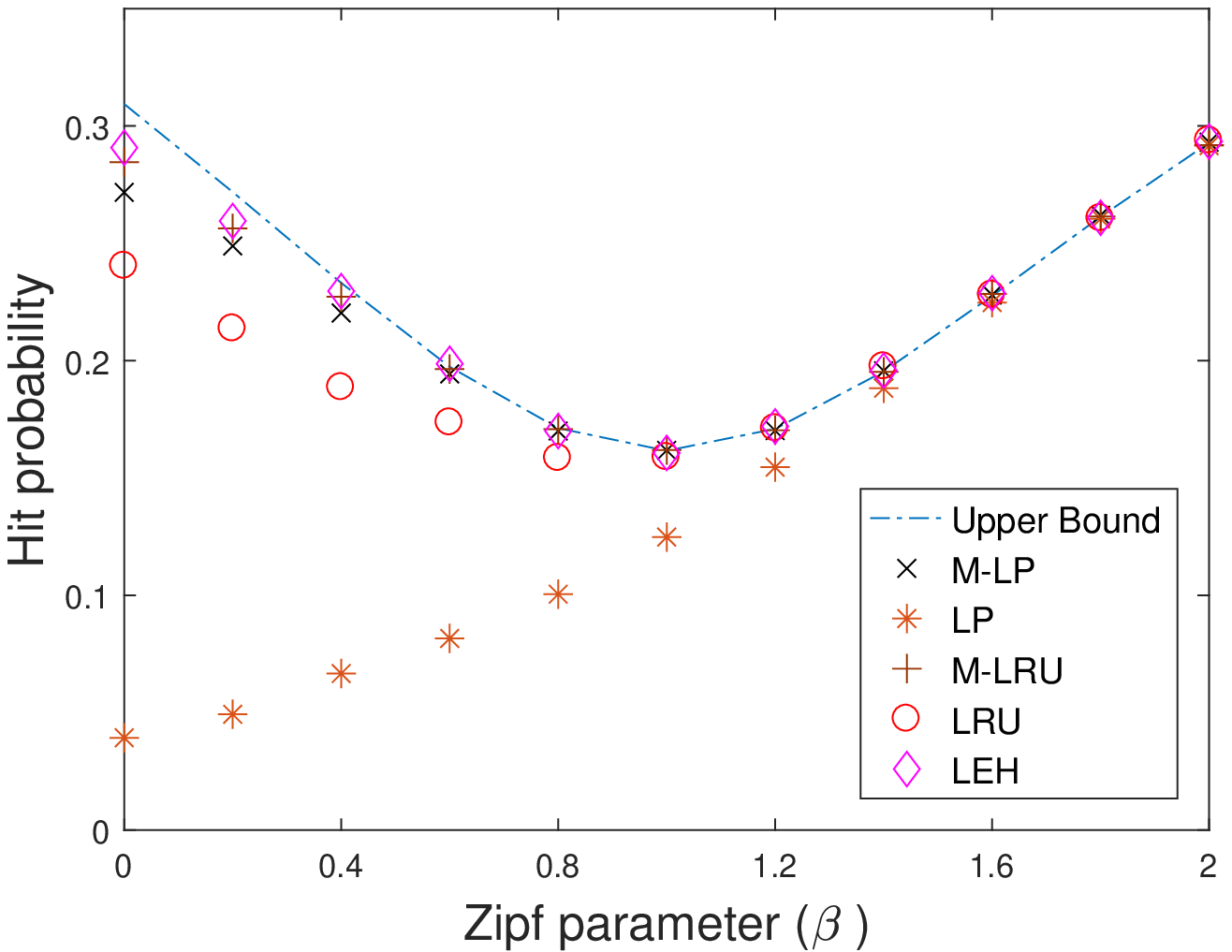}
 \caption*{(a)}
\end{minipage}
\hfill
\begin{minipage}[t]{0.5\linewidth}
\centering
\includegraphics[width=\linewidth]{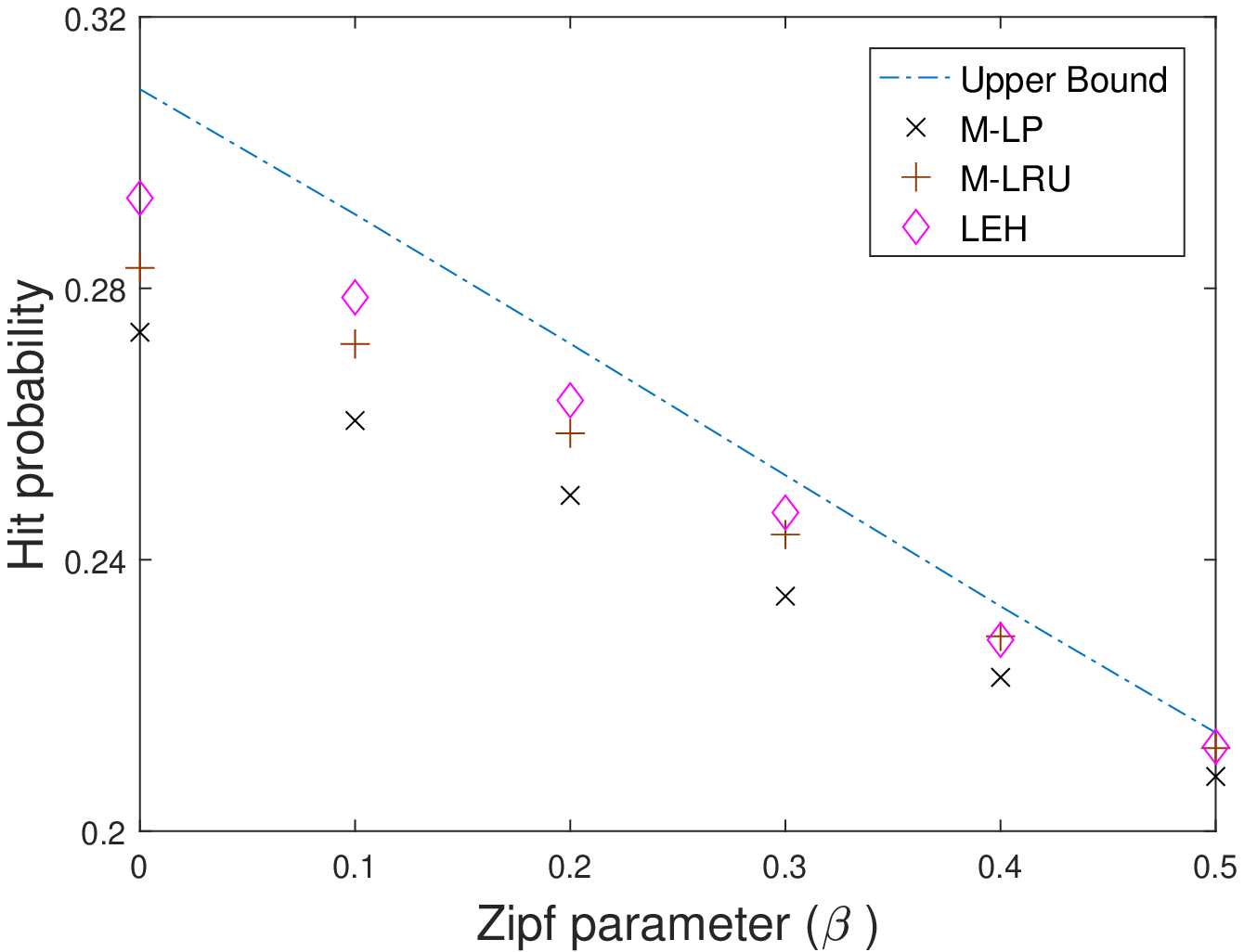}
 \caption*{(b)}
\end{minipage}
 \caption{Performance of caching policies as a function of Zipf parameter $\beta$ for $ n=100$, $m=30 $ and $F(i)=1+i, \; i \in \lbrace 1,2,\ldots n \rbrace $.}
\label{varybetafull2} 
\end{figure*}

\begin{figure}[h!]
\centering
\includegraphics[width=0.5\linewidth]{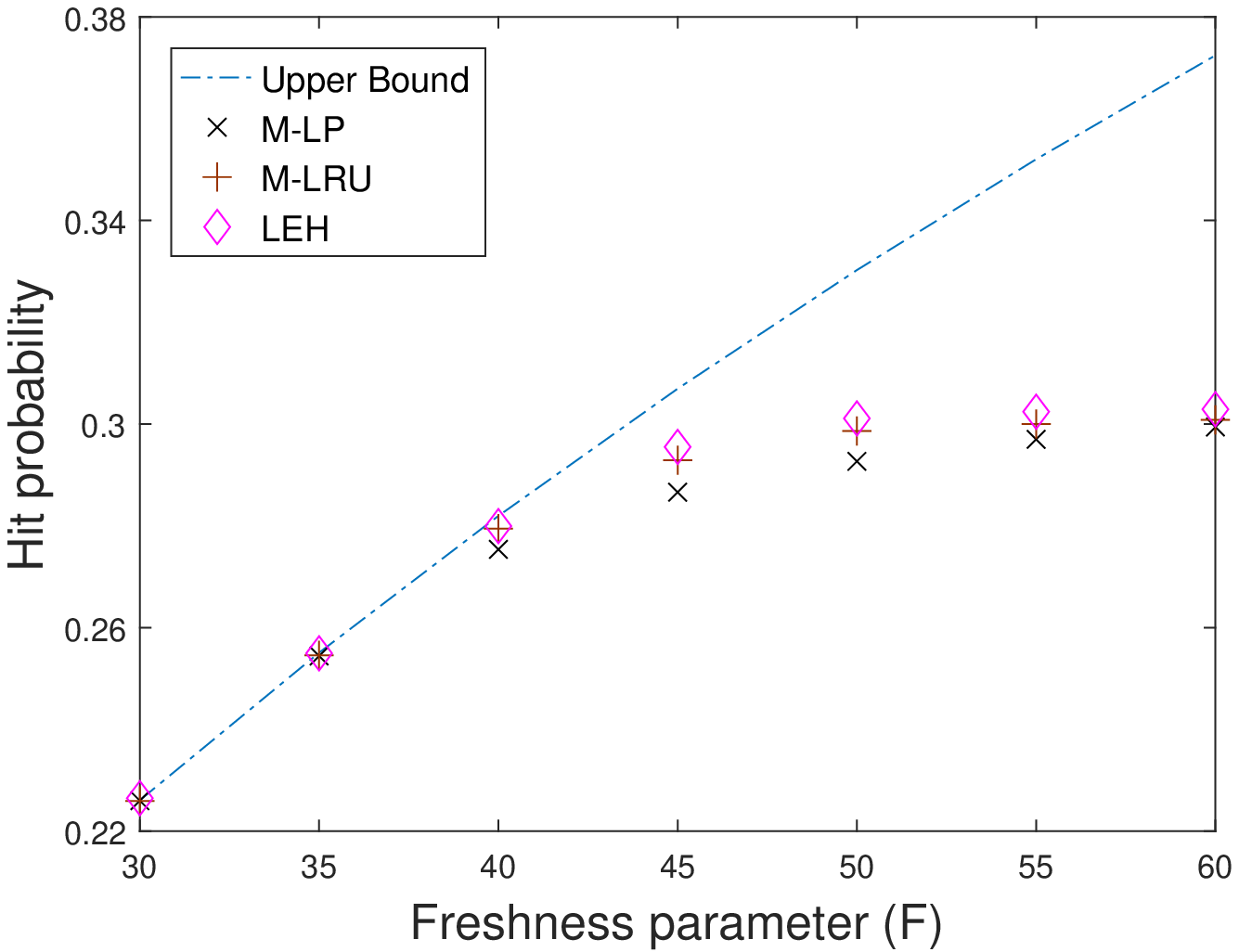}
 \caption{Performance of caching policies as a function of freshness specification $F$ for $ n=100$, $m=30$ and $\beta=0.1 $.}
\label{varyF} 
\end{figure}

\begin{figure*}
\begin{minipage}[b]{0.5\linewidth}
\centering
\includegraphics[width=\linewidth]{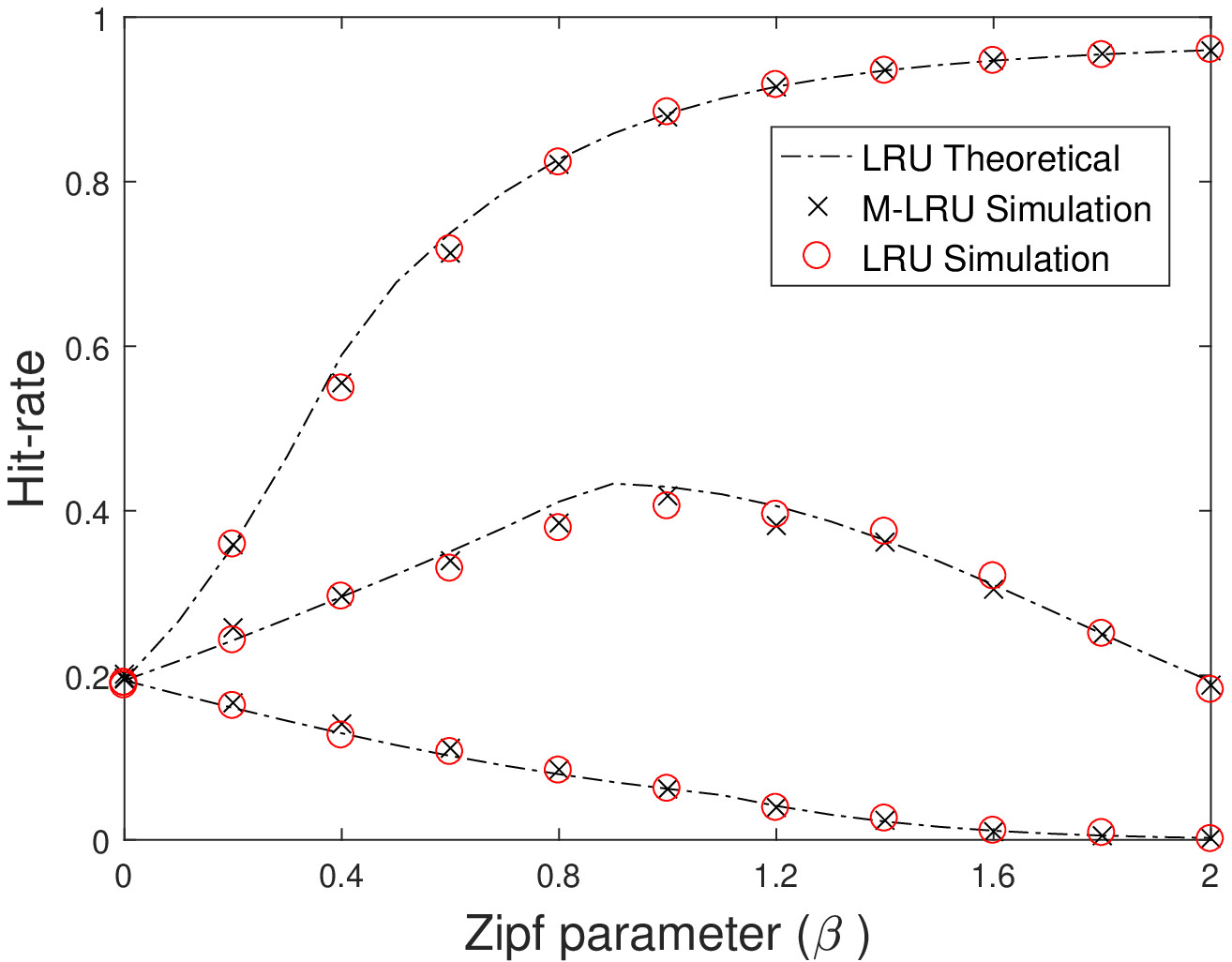}
 \caption*{\text{(a})}
\end{minipage}
\begin{minipage}[b]{0.5\linewidth}
\includegraphics[width=\linewidth]{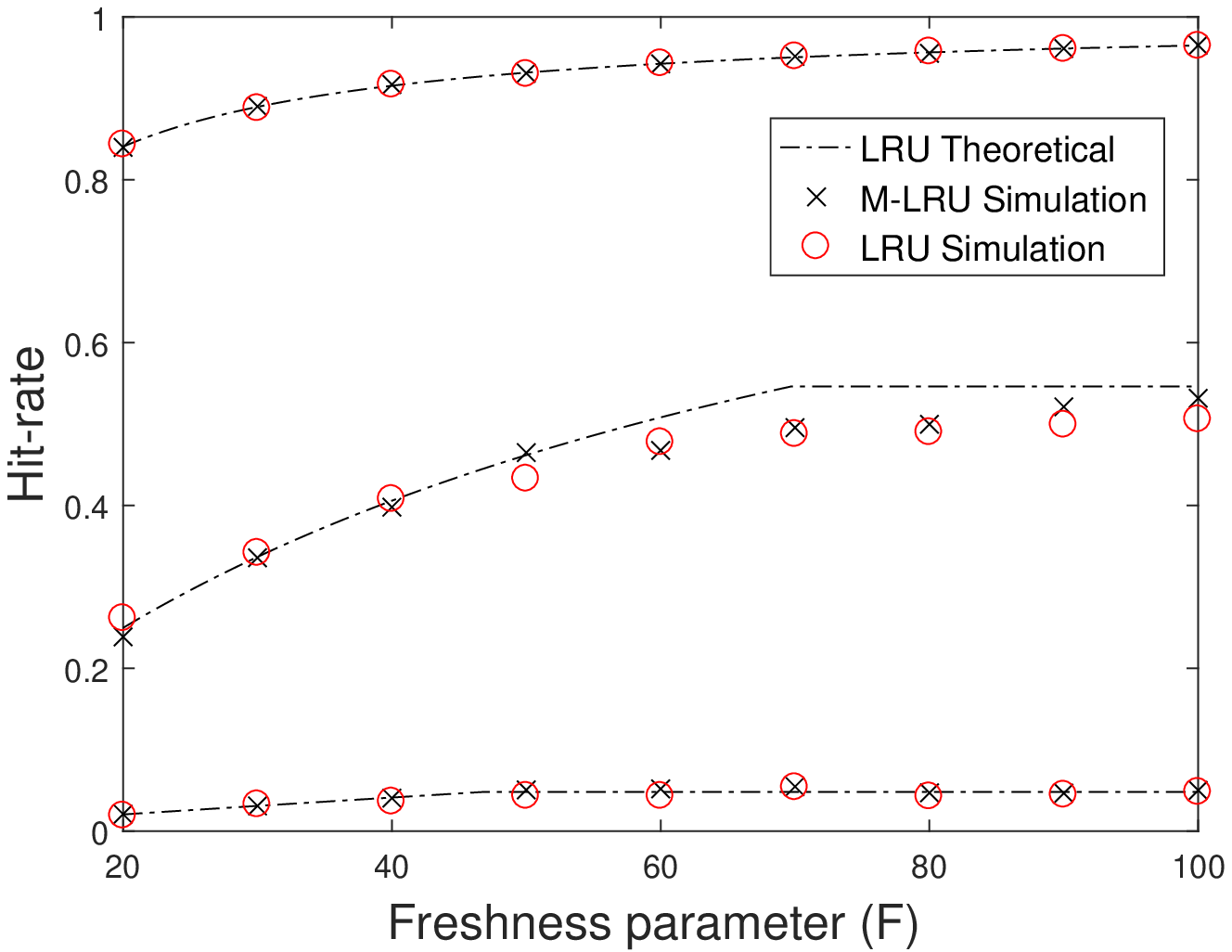}
\caption*{(b)}
\end{minipage}
\caption{Hit-rate against (a) Zipf parameter $\beta$ for contents 1, 10 and 100. $n=100$, $m=20$ and $F=40$ and (b) freshness parameter $F$ for contents 1, 10 and 100. $n=100$, $m=20$ and $\beta=1.2.$} 
\label{hits_varybeta_varyF}
\end{figure*} 

\section{Concluding Remarks}  \label{conclusion}
In this paper, we studied caching policies under the setting in which the users impose freshness constraints on the requested content. First, we quantified the optimal performance that a caching policy may achieve, subject to these constraints, in the form of of a universal upper bound. We then obtained content-wise hit-rates for the LP policy. For a practical scenario with Zipf distributed requests, we proved that as the number of contents ($n$) increases, the LP policy asymptotically attains the optimal performance as long as the cache-size also increases with $n$. Next, we obtained an accurate approximation for the LRU content-wise hit-rates for large $n$. To achieve this, we associated the problem of estimating the \emph{characteristic time} of a content in the LRU policy with the classical Coupon Collector's Problem. We provided analytical results in the form of tight concentration bounds on the characteristic time to justify the accuracy of our approximations. Further, we improved the LP and LRU policies by exploiting the knowledge of the freshness specifications to eject \emph{stale} data from the cache and thereby removed redundancies that inhibit cache performance. Finally, we proposed a probabilistic algorithm (LEH policy) which enabled us to prioritize the contents by accounting for both their freshness specifications and popularities. We verified through extensive simulations that this algorithm indeed improves the caching performance and fairs better than the other policies considered.

\iftoggle{SINGLE_COL}{\vspace{-1.0cm}}{}


\end{document}